\documentclass[a4paper,UKenglish,cleveref, autoref, thm-restate]{lipics-v2021}

\usepackage{tikz}
\usetikzlibrary{automata,calc, positioning, graphs, decorations.pathmorphing,shapes,arrows.meta,arrows,shapes.misc, fit, math}

\usepackage[colorinlistoftodos,prependcaption,textsize=tiny]{todonotes}
\usepackage{xargs}
\newcommandx{\unsure}[2][1=]{\todo[linecolor=green,backgroundcolor=green!25,bordercolor=green,#1]{\normalsize #2}}

\newcommand{\eps}{\varepsilon}

\newcommand{\qeds}{$\blacksquare$}

\newcommand{\F}{\mathbb{F}}

\renewcommand{\S}{\mathcal{S}}

\DeclareMathOperator{\col}{col}
\DeclareMathOperator{\rank}{rank}
\newcommand{\cw}{{\operatorname{ctw}}} 
\newcommand{\ctw}{{\operatorname{ctw}}}
\newcommand{\tw}{{\operatorname{tw}}}
\newcommand{\pw}{{\operatorname{pw}}}
\DeclareMathOperator{\supp}{supp}

\pdfoutput=1 
\hideLIPIcs  

\title{Tight bounds for counting colorings and connected edge sets parameterized by cutwidth\footnote{All authors are supported by the project CRACKNP that has received funding from the European Research Council (ERC) under the European Union's Horizon 2020 research and 	innovation programme (grant agreement No 853234).}}

\author{Carla Groenland}{Utrecht University, The
    Netherlands }{c.e.groenland@uu.nl}{orcid}{}

\author{Isja Mannens}{Utrecht University, The
    Netherlands }{i.m.e.mannens@uu.nl}{orcid}{}
\author{Jesper Nederlof}{Utrecht University, The
    Netherlands }{j.nederlof@uu.nl}{orcid}{}
\author{Krisztina Szil\'agyi}{Utrecht University, The
    Netherlands }{k.szilagyi@uu.nl}{orcid}{}

\authorrunning{C. Groenland, I. Mannens, J. Nederlof and K. Szil\'agyi} 

\Copyright{Carla Groenland, Isja Mannens, Jesper Nederlof and Krisztina Szil\'agyi} 

\ccsdesc[500]{Theory of computation~Parameterized complexity and exact algorithms}

\keywords{connected edge sets, cutwidth, parameterized algorithms, colorings, counting modulo $p$} 

\category{} 

\relatedversion{} 

\acknowledgements{We thank the anonymous reviews for their detailed comments.}

\nolinenumbers 

\EventEditors{John Q. Open and Joan R. Access}
\EventNoEds{2}
\EventLongTitle{42nd Conference on Very Important Topics (CVIT 2016)}
\EventShortTitle{CVIT 2016}
\EventAcronym{CVIT}
\EventYear{2016}
\EventDate{December 24--27, 2016}
\EventLocation{Little Whinging, United Kingdom}
\EventLogo{}
\SeriesVolume{42}
\ArticleNo{23}

\begin{document}

\maketitle

\begin{abstract}
We study the fine-grained complexity of counting the number of colorings and connected spanning edge sets parameterized by the cutwidth and treewidth of the graph. While decompositions of small treewidth decompose the graph with small vertex separators, decompositions with small cutwidth decompose the graph with small \emph{edge} separators.

Let $p,q \in \mathbb{N}$ such that $p$ is a prime and $q \geq 3$. We show:
\begin{itemize}
	\item If $p$ divides $q-1$, there is a $(q-1)^{\cw}n^{O(1)}$ time algorithm for counting list $q$-colorings modulo $p$ of $n$-vertex graphs of cutwidth $\textup{ctw}$. Furthermore, there is no $\varepsilon>0$ for which there is a $(q-1-\varepsilon)^{\textup{ctw}} n^{O(1)}$ time algorithm that counts the number of list $q$-colorings modulo $p$ of $n$-vertex graphs of cutwidth ${\textup{ctw}}$, assuming the Strong Exponential Time Hypothesis (SETH).
	\item If $p$ does not divide $q-1$, there is no $\varepsilon>0$ for which there exists a $(q-\varepsilon)^{\textup{ctw}} n^{O(1)}$ time algorithm that counts the number of list $q$-colorings modulo $p$ of $n$-vertex graphs of cutwidth ${\textup{ctw}}$, assuming SETH.
\end{itemize}
The lower bounds are in stark contrast with the existing $2^{\ctw}n^{O(1)}$ time algorithm to compute the chromatic number of a graph by Jansen and Nederlof~[Theor. Comput. Sci.'18].

Furthermore, by building upon the above lower bounds, we obtain the following lower bound for counting connected spanning edge sets: there is no $\varepsilon>0$ for which there is an algorithm that, given a graph $G$ and a cutwidth ordering of cutwidth $\cw$, counts the number of spanning connected edge sets of $G$ modulo $p$ in time $(p - \varepsilon)^\ctw n^{O(1)}$, assuming SETH. We also give an algorithm with matching running time for this problem.

Before our work, even for the treewidth parameterization, the best conditional lower bound by Dell et al.~[ACM Trans. Algorithms'14] only excluded $2^{o({\textup{tw}})}n^{O(1)}$ time algorithms for this problem.

Both our algorithms and lower bounds employ use of the matrix rank method, by relating the complexity of the problem to the rank of a certain `compatibility matrix' in a non-trivial way. 
\begin{picture}(0,0)
\put(-20,-250)
{\hbox{\includegraphics[width=40px]{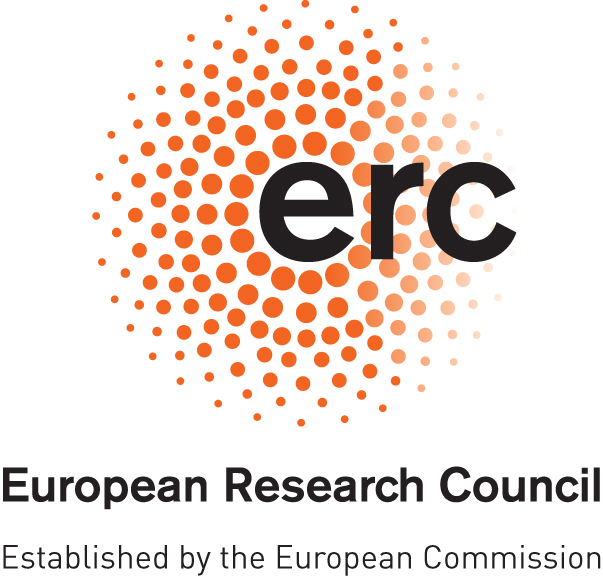}}}
\put(-30,-240)
{\hbox{\includegraphics[width=60px]{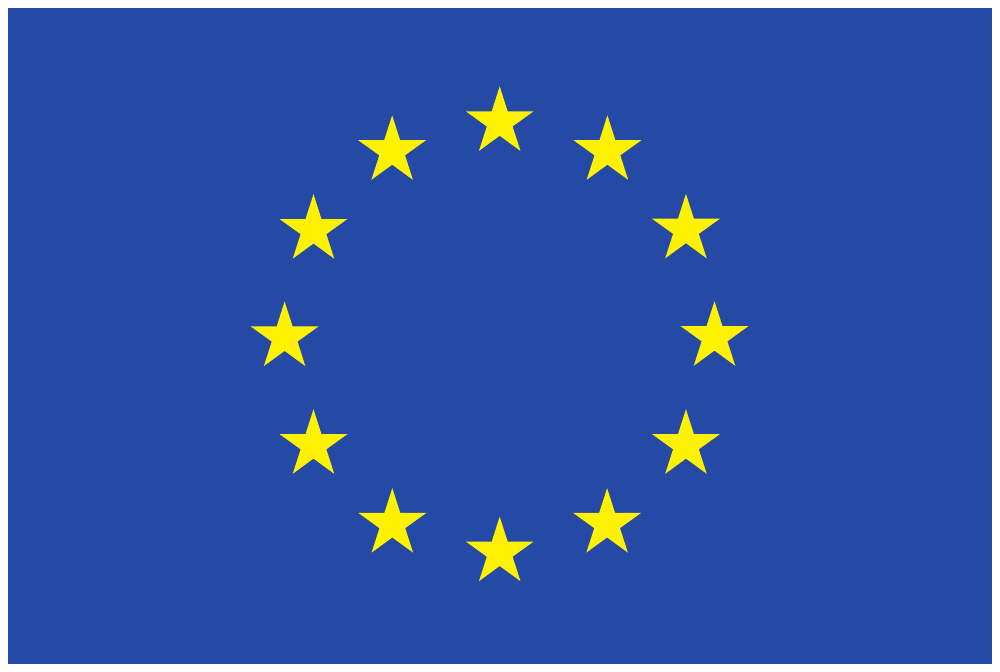}}}
\end{picture}
\end{abstract}
\newpage
\setcounter{page}{1}
\section{Introduction}
A popular topic of interest in (fine-grained) algorithmic research is to determine the decomposability of NP-hard problems in easier subproblems.
A natural decomposition strategy is often implied by decomposing the solution into sub-solutions induced by a given decomposition of the input graph such as tree decompositions, path decompositions, or tree depth decompositions, independent of the problem to be solved. However, the efficiency of such a decomposition can wildly vary per computational problem. Recently, researchers developed tools that allow them to get a precise understanding of this efficiency : non-trivial algorithmic tools (such as convolutions and the cut-and-count method~\cite{cygan2011solving, RooijBR09}) were developed to give algorithms that have an optimal running time conditioned on hypotheses such as the Strong Exponential Time Hypothesis (SETH)~\cite{ImpagliazzoP01}.
While the efficiency of such decompositions has been settled for most decision problems parameterized by treewidth, many other interesting settings remain elusive.
Two of them are \emph{cutwidth} and \emph{counting problems}.

The \emph{cutwidth} of an ordering of the vertices of the graph is defined as the maximum number of edges with exactly one endpoint in a prefix of the ordering (where the maximum is taken over all prefixes of the ordering). The \emph{cutwidth} of a graph is defined to be the minimum width over all its cutwidth orderings. Cutwidth is very similar to pathwidth, except that cutwidth measures the number of edges of a cut, while the pathwidth measures the number of \emph{endpoints} of edges over the cut.
Thus the cutwidth of a graph is always larger than its pathwidth. But for some problems a decomposition scheme associated with a cutwidth ordering of cutwidth $k$ can be used much more efficiently than a decomposition of pathwidth $k$.
A recent example of such a problem is the $q$-coloring problem:\footnote{Recall that a $q$-coloring is a mapping from the vertices of the graph to $\{1,\ldots,q\}$ such that every two adjacent vertices receive distinct colors, and the $q$-coloring problem asks whether a $q$-coloring exists.} While there is a $(q-\varepsilon)^\pw$ lower bound~\cite{LokshtanovMS18} assuming SETH, there is a $2^\ctw n^{O(1)}$ time randomized algorithm~\cite{jansen2019computing}.

\emph{Counting problems} pose an interesting challenge if we want to study their decomposability.
Counting problems are naturally motivated if we are interested in any statistic rather than just existence of the solutions space. While often a counting problem behaves very similarly to its decision version (as in, the dynamic programming approach can be fairly directly extended to solve the counting version as well), for some problems there is a rather puzzling increase in complexity when going from the decision version to the counting version.
\footnote{Two examples herein are detecting/counting perfect matchings (while the decision version is in $P$, the counting version can not be solved in time $(2-\eps)^\tw n^{O(1)}$ for any $\eps >0$ assuming the SETH~\cite{CurticapeanM16}) and Hamiltonian cycles (while the decision version can be solved in $(2+\sqrt{2})^\pw$ time~\cite{CyganKN18}, the counting version can not be solved in time $(6-\eps)^\tw n^{O(1)}$ for any $\eps >0$ assuming the SETH~\cite{CurticapeanLN18}).}

One of the most central problems in counting complexity is the evaluation of the \emph{Tutte polynomial}. The strength of this polynomial is that it expresses all graph invariants that can be written as a linear recurrence using only the edge deletion and contraction operation \cite{oxley1979tutte}, and its evaluations specialize to a diverse set of parameters ranging from the number of forests, nowhere-0 flows, $q$-colorings and spanning connected edge sets.

An interesting subdirection within counting complexity that is in between the decision and counting version and that we will also address in this paper is \emph{modular counting}, where we want to count the number of solutions modulo a number $p$. This is an interesting direction since the complexity of the problem at hand can wildly vary for different $p$ (see~\cite{Valiant06} for a famous example), but in the setting of this paper it is also naturally motivated: For example, the cut-and-count method achieves the fastest algorithms for several decision problems by actually solving the modular counting variant instead.

\subsection{Our results}
In this paper we study the complexity of two natural hard (modular) counting problems: Counting the number of $q$-colorings of a graph and counting the number of spanning connected edge sets, parameterized by the cutwidth of the graph.
\subparagraph{Counting Colorings}
Let $G$ be a graph and suppose that for each $v \in V$ we have an associated list $L(v) \subseteq \{1,\ldots,q\}$. A \emph{list $q$-coloring} is a coloring $c$ of $G$ such that $c(v) \in L(v)$ for each $v \in V$. Two colorings are \emph{essentially distinct} if they cannot be obtained from each other by permuting the color classes. Since the number of essentially distinct colorings is $q!$ times the number of distinct colorings (assuming the chromatic number of the graph is $q$), counting colorings modulo $p$ may become trivial if $p \leq  q$. For this reason, we focus on counting essentially distinct colorings in our lower bounds.

In this paper, we will focus on counting list $q$-colorings modulo a prime number $p$.
Our main theorem reads as follows:

\begin{theorem}
	\label{thm:main}
Let $p,q \in \mathbb{N}$ with $p$ prime and $q \geq 3$. 
\begin{itemize}
	\item If $p$ divides $q-1$, then there is a $(q-1)^{\cw}n^{O(1)}$ time algorithm for counting list $q$-colorings modulo $p$ of $n$-vertex graphs of cutwidth $\cw$. Furthermore, there is no $\varepsilon>0$ for which there exists a $(q-1-\varepsilon)^\cw n^{O(1)}$ time algorithm that counts the number of essentially distinct $q$-colorings modulo $p$ in time $(q-1-\varepsilon)^\cw n^{O(1)}$, assuming SETH.
	\item If $p$ does not divide $q-1$, there is no $\varepsilon>0$ for which there exists a $(q-\varepsilon)^\cw n^{O(1)}$ time algorithm that counts the number of essentially distinct $q$-colorings modulo $p$, assuming SETH.
\end{itemize}
\end{theorem}

Thus, we show that under the cutwidth parameterization, the (modular) counting variant of $q$-coloring is much harder than the decision, as the latter can be solved in $2^\cw n^{O(1)}$ time with a randomized algorithm~\cite{jansen2019computing}.
Additionally, we show there is a curious jump in complexity based on whether $p$ divides $q-1$ or not: Since our bounds are tight, this jump is inherent to the problem and not an artifact of our proof.

The proof strategy of all items of Theorem~\ref{thm:main} relates the complexity of the problems to a certain \emph{compatibility matrix}. This is a Boolean matrix that has its rows and columns indexed by partial solutions, and has a $1$ if and only if the corresponding partial solutions combine into a global solution.
In previous work, it was shown that the rank of this matrix can be used to design both algorithms~\cite{BodlaenderCKN15,CyganKN18,jansen2019computing, Nederlof20} and lower bounds~\cite{CurticapeanLN18,CyganKN18}.
 
With this in mind, the curious jump can intuitively be explained as follows. Consider the base case where the graph is a single edge and we decompose a (list) $q$-coloring into the two colorings induced on the vertices. The compatibility matrix corresponding to this decomposition is the complement of an $q \times q$ identity matrix. This matrix has full rank if $p$ does not divide $q-1$ and it has rank $q-1$ otherwise. We believe this is a very clean illustration of the rank based methods, since it explains a curious gap that would be rather mysterious without the rank based perspective.

\subparagraph{Connected Spanning Edge Sets and Tutte polynomial}
 We say that $X\subseteq E$ is a \emph{connected spanning edge set} if $G[X]$ is connected and every vertex is adjacent to an edge in $X$. Our second result is about counting the number of such sets. This problem is naturally motivated: It gives the probability that a random subgraph remains connected, and is an important special case of the Tutte polynomial.
 We determine the complexity of counting connected spanning edge sets by treewidth and cutwidth by giving matching lower and upper bounds:
\begin{theorem} 
\label{thm:main_upperboundcse}
\label{thm:main_lbcse}
Let $p$ be a prime number. There is an algorithm that counts the number of connected edge sets modulo $p$ of $n$-vertex graphs of treewidth $\tw$ in time $p^{\tw}n^{O(1)}$.

Furthermore, there is no $\eps >0$ for which there is an algorithm that counts the number of spanning connected edge sets modulo $p$ of $n$-vertex graphs of cutwidth $\cw$ in time $(p - \epsilon)^\ctw n^{O(1)}$, assuming SETH.
\end{theorem}

Note that before our work, even for the treewidth parameterization, the best conditional lower bound by Dell et al.~\cite{DellHMTW14} only excluded $2^{o(\tw)}n^{O(1)}$ time algorithms for this problem.

While the algorithm follows relatively quickly by using a cut-and-count type dynamic programming approach, obtaining the lower bound is much harder.

In fact, for related counting variants of connectivity problems such as counting the number of Hamiltonian cycles or Steiner trees, $2^{O(\tw)}n^{O(1)}$ time algorithms do exist. So one may think that connected spanning edge sets can be counted in a similar time bound. But in Theorem~\ref{thm:main_upperboundcse} we show that this is not the case (by choosing $p$ arbitrarily large). 

To prove the lower bound, we make use of an existing  formula for the Tutte polynomial that relates the number of connected spanning edge sets to the number of essentially distinct colorings, and subsequently apply Theorem~\ref{thm:main}.

\subparagraph{Organization} The rest of the paper is organized as follows: in Section \ref{sec:prelim} we introduce the notation that will be used throughout the paper and define the color compatibility matrix. In Section \ref{sec:colupp} we prove the upper bound for \#$q$-coloring modulo $p$. Section \ref{sec:collow} contains the results about lower bounds. We conclude the paper by discussing directions for further research. The appendix contains the proofs omitted from previous sections.





\subsection{Related work}


\subparagraph{Coloring} Counting the number of colorings of a graph is known to be $\#P$-complete, even for special classes of graphs such as triangle free regular graphs \cite{greenhill2000complexity}. Bj\"orklund and Husfeldt \cite{bjorklund2006inclusion} and Koivisto~\cite{Koivisto06} gave a $2^nn^{O(1)}$ algorithm for counting $q$-colorings, and a more general $2^nn^{O(1)}$ time algorithm even evaluates any point of the Tutte polynomial~\cite{BjorklundHKK08}.

A $q$-coloring of a graph $G$ is a special case of \emph{$H$-coloring}, i.e. a homomorphism from $G$ to a given graph $H$. Namely, $q$-colorings  correspond to homomorphisms from $G$ to $K_q$, i.e. $K_q$-colorings.  
Dyer and Greenhill \cite{dyer2000complexity} showed that counting the number of $H$-colorings is $\#P$-complete unless $H$ is one of the few exceptions (an independent set, a complete graph with loops on every vertex or a complete bipartite graph). Kazeminia and Bulatov \cite{kazeminia2019counting} classified the hardness of counting $H$-colorings modulo a prime $p$ for square-free graphs $H$.

\subparagraph{Methods}
Our approach makes use of the rank based method, and in particular the so-called color compatibility matrix introduced in~\cite{jansen2019computing}. This matrix tells us whether we can `combine' two colorings. In~\cite{jansen2019computing}, the authors studied the rank of a different matrix with the same support as the color compatibility matrix, whereas in this paper we use the rank directly. The rank based method has been used before only once for an algorithm for a counting problem in~\cite{CyganKN18} and only once for a lower bound for a counting problem in~\cite{CurticapeanLN18}. 

The Tutte polynomial $T(G;x,y)$ is a graph polynomial in two variables which describes how $G$ is connected. In particular, calculating $T(G;x,y)$ at specific points gives us the number of subgraphs of $G$ with certain properties: $T(G;2,1)$ is equal to the number of forests in $G$, $T(G;1,1)$ is the number of spanning forests, $T(G;1,2)$ counts the number of spanning connected subgraphs etc. We will use the properties of the Tutte polynomial to give a lower bound on the complexity of counting spanning connected edge sets.

\section{Preliminaries}
\label{sec:prelim}
In this section, we introduce the notation that will be used throughout the paper.  

\subsection{Notation and standard definitions}
For integers $a,b$, we write $[a,b]=\{a,a+1,\dots ,b\}$ for the integers between $a$ and $b$, and for a natural number $n$ we short-cut $[n]=[1,n]=\{1,\dots,n\}$. Throughout the paper, $p$ will denote a prime number and $\mathbb{F}_p$ the finite field of order $p$. We will use $a\equiv_p b$ to denote that $a$ and $b$ are congruent modulo $p$, i.e. that $p$ divides $a-b$. We write $\mathbb{N}=\{1,2,\dots\}$ for the set of natural numbers.

For a function $f:A\rightarrow \mathbb{Z}$ (where $A$ is any set),  we define the support of $f$ as the set $\supp(f)=\{a\in A\: : \: f(a)\neq 0 \}$. For $B\subseteq A$, the function $f|_B:B\to \mathbb{Z}$ is defined as $f|_B(b)=f(b)$ for all $b\in B$.

In this paper, all graphs will be undirected and simple. Given a graph $G=(V,E)$ and a vertex $v\in V$, we denote by $N(v)$ the open neighbourhood of $v$, i.e. the set of all vertices adjacent to $v$. We often use $n$ for the number of vertices of $G$, and denote the cutwidth of $G$ by $\cw$. We sometimes write $V(G)$ for the vertex set of the graph $G$.

Note that, if $G$ is not connected, we can count the number of $q$-colorings in each connected component and multiply them to get the total number of $q$-colorings of $G$. Therefore, we may assume that $G$ is connected.

Given a graph $G=(V,E)$, and lists $L:V\rightarrow 2^{[q]}$, a \emph{list $q$-coloring} of $G$ is a coloring $c:V\rightarrow [q]$ of its vertices such that $c(u)\neq c(v)$ for all edges $uv$ and $c(v)\in L(v)$ for all vertices $v$. We will often abbreviate `list $q$-coloring' to `coloring'. For a subset $B \subseteq V(G)$, we will use the abbreviation $c(B) = \{c(v) : v \in B\}$.

Cutwidth and treewidth are graph parameters which are often used in parameterized complexity. Informally, treewidth describes how far a graph is from being a tree. 
The cutwidth is defined as follows.
\begin{definition}
	The \emph{cutwidth} of a graph $G$ is the smallest $k$ such that its vertices can be arranged in a sequence $v_1,\dots,v_n$ such that for every $i\in [n-1]$, there are at most $k$ edges between $\{v_1,\dots,v_i\}$ and $\{v_{i+1},\dots,v_n\}$.
\end{definition}

We recall the definition of Tutte polynomial.
\begin{definition}
	For a graph $G$, we denote by $T(G;x,y)$ the \emph{Tutte polynomial} of $G$ evaluated at the point $(x,y)$. If $G$ has no edges we have $T(G;x,y) = 1$. Otherwise we have
	\[
	T(G;x,y) = \sum \limits_{A \subseteq E(G)} (x-1)^{r(E) - r(A)} (y-1)^{|A|-r(A)}
	\]
	where $r(A) = |V(G)| - k(A)$ indicates the rank of the edge set $A$ and $k(A)$ indicates the number of connected components of $(V,A)$.
\end{definition}

Note that $T(G;1,2)$ is exactly the number of spanning connected edge sets. 

We denote the counting version of a problem by using the prefix \#, and the counting modulo $p$ version of by using $\#_p$ (e.g. $\#_p \textsc{SAT}$, $\#_p \textsc{CSP}$).  

\subsection{The color compatibility matrix and its rank}
\label{subsec:rank}
Given a subset $A\subseteq V$, we use $\col_L(A)$ to denote the set of all list $q$-colorings of $G[A]$. If it is clear which lists are used, we omit the subscript. 

It is often useful to color parts of the graph separately, and then `combine' those colorings. If two colorings can be combined without conflicts, we call them compatible:

\begin{definition}
	\label{def:ccrelation}
	For subsets $A,B\subseteq V$ and colorings $x\in \col(A)$, $z\in \col(B)$, we say that $x$ and $z$ are \emph{compatible}, written $x\sim z$, if
	\begin{itemize}
		\item $x(v)=z(v)$ for all $v\in A \cap B$, and
		\item $x(u)\neq z(v)$ for all $uv\in E$, where $u\in A$ and $v\in B$.
	\end{itemize}
For a set of colorings $\S\subseteq \col(B)$, we write $\S[x]$ for the set of colorings $y\in \S$ that are compatible with $x$. 
\end{definition}
If $x\sim z$, then we define $x\cup z$ as the $q$-list coloring of $G[A\cup B]$ such that $(x\cup z)(a)=x(a)$ for all $a\in A$ and $(x\cup z)(b)=z(b)$ for all $b\in B$. This is well-defined by the definition above.

A key definition for this paper is the following.
\begin{definition}
	\label{def:ccmatrix}
Let $(X\cup Y, E)$ be a bipartite graph and $q$ a natural number. 
The $q$th \emph{color compatibility matrix} $M$ is indexed by all $q$-colorings of $X$ and $Y$, with 
\[
M[x,y]=\begin{cases}
	1, & \text{ if } x\sim y, \\
	0, & \text{ otherwise,}
	\end{cases}
\]
for $x\in \col(X)$ and $y\in \col(Y)$.
\end{definition}
We denote the color compatibility matrix indexed by all $q$-colorings associated with the bipartite graph that is matching on $t$ vertices by $J_t$, and short-hand $J:=J_1$.

We will show that, if $p$ divides $q-1$, we can count all $q$-list colorings modulo $p$ more quickly due to the following bound on the rank of the color compatibility matrices.
\begin{lemma} 
\label{lem:rankComp}
	Let $p$ be a prime, $q$ a natural number and let $G=(X\cup Y, E)$ be a bipartite graph with $q$th color compatibility matrix $M$. Then the rank of $M$ over $\mathbb{F}_p$ satisfies
	\[
	\rank_p(M)\leq \begin{cases}
		(q-1)^{|E|} & \text{ if } p\text{ divides }q-1,\\
		q^{|E|} & \text{ otherwise.} 
	\end{cases}\]
	Moreover, equality is achieved if $G$ is a perfect matching.
\end{lemma}
\begin{proof}
    	Let $p$ be a prime, $q$ a natural number and let $G=(X\cup Y, E)$ be a bipartite graph with $q$th color compatibility matrix $M$. We need to show that the rank of $M$ over $\mathbb{F}_p$ satisfies
	\[
	\rank_p(M)\leq \begin{cases}
		(q-1)^{|E|} & \text{ if } p\text{ divides }q-1,\\
		q^{|E|} & \text{ otherwise,} 
	\end{cases}\]
	with equality achieved if $G$ is a matching.
	
In the following proof, we will omit the subscript from $\rank_p$.
Recall that we denote the color compatibility matrix associated with the bipartite graph that is a matching on $t$ pairs by $J_t$.
We will first show that $\rank(J_t)=\rank(J_1)^t$. Then we will show that $\rank(M)\leq \rank(J_{|E|})$ for any graph $G$ and its color compatibility matrix $M$. Finally, we compute $\rank(J_1)$, which will cause the difference between the cases when $p$ divides $q-1$ and when it does not. 
	
	\begin{itemize}
	\item Let us consider the perfect matching case. 
	We will show by induction on $r$ that $M=J_r$ is the $r$th Kronecker power of $J_1$, which implies $\rank(J_r)=\rank(J_1)^r$.
	
	The Kronecker product of an $m\times n$ matrix $A$ and an $p\times q$ matrix $B$ is defined as the following $pm\times qn$ block matrix 
	\[
	A\otimes B=\begin{pmatrix}
		a_{11}B & \cdots & a_{1n}B \\
		\vdots & \ddots & \vdots \\
		a_{m_1}B & \cdots & a_{mn}B
		\end{pmatrix}.
		\]
The entries on the diagonal of $J_1$ are zero and the off-diagonal entries are one: two colors are in conflict if and only if they are the same. Consider now some $r\geq 2$ and let $X=\{x_1,\dots,x_r\}$, $Y=\{y_1,\dots,y_r\}$ and $E=\{\{x_1,y_1\},\dots,\{x_r,y_r\}\}$.
	We write a coloring of $X$ as $c_1\dots c_r$, where $c_i$ is the color assigned to $x_i$. We assume that the entries in $M$ are indexed by (lexicographically ordered) colorings $c_1\dots c_r$ and $c'_1\dots c'_r$, where $c_i$ and $c'_i$ are the colors of $x_i$ and $y_i$ respectively. We will split up $M$ into $q^2$ $(q^{r-1}\times q^{r-1})$-matrices:
		\[
	M=\begin{pmatrix}
		M^{(11)} & \dots & M^{(1q)} \\
		\vdots & \ddots & \vdots \\
		M^{(q1)} & \dots & M^{(qq)}
	\end{pmatrix}.
	\]
	For $i,j\in [q]$, the matrix $M^{(ij)}$ is indexed by colorings where $x_1$ receives color $i$ and $y_1$ receives color $j$. Therefore all $M^{(ii)}$ equal the zero matrix. For $i\neq j$, $M^{(ij)}$ is the color compatibility matrix $J_{r-1}$ of $((X\setminus \{x_1\})\cup (Y\setminus \{y_1\}), E\setminus \{x_1,y_1\})$. This proves that $M$ equals the Kronecker product $J_{1}\otimes J_{r-1}$.
	
	Therefore, $J_r$ is the $r$th Kronecker power of $J_1$, so it has rank $(\rank(J_1))^r$.
	
	\item Let us now consider the general case. Let $G$ be any bipartite graph with $t\geq 1$ edges and let $M$ the corresponding color compatibility matrix. By a similar argument to the one above, we may assume that $G$ has no isolated vertices (else we can write $M$ as the Kronecker product of the all ones matrix and the color compatibility matrix of the graph $G$ from which one isolated vertex is removed).
	Let $G_t$ be a perfect matching with $t$ edges and $J_t$ the corresponding color compatibility matrix. We claim that $M$ is a submatrix of $J_t$. 
	
	Namely, note that we can obtain $G$ from $G_t$ by identifying some of its vertices. We observe how the color compatibility matrix changes after identifying two vertices: identifying vertices $u,v$ from the left-hand side of the bipartition corresponds to deleting all rows where $u$ and $v$ are assigned a different color. Similarly, identifying vertices in the right part corresponds to deleting columns. Therefore, $M$ can be obtained from $J_t$ by deleting some rows and columns. In particular, $\rank(M)\leq \rank(J_t)$. 
	
	\item Finally, we compute the rank of $J_1$. Recall that $J_1$ is a $q\times q$ matrix with zeros on the diagonal and ones off-diagonal. If $p$ divides $q-1$, this matrix has rank $q-1$ and if not, it has full rank. 
	\end{itemize}
	Therefore, $\rank(J_t)$ is equal to $(q-1)^t$ and $q^t$ respectively. The rank of a color compatibility matrix of an arbitrary graph $G$ is bounded by $\rank(J_{|E|})$, giving the required bounds. 
\end{proof}
In particular, $J_t$ is invertible mod $p$ if and only if $p$ does not divide $q-1$.

\section{Algorithm for \#$q$-coloring modulo $p$}
\label{sec:colupp}
In this section we prove the first part of the first item of Theorem~\ref{thm:main}:
\begin{theorem}
	\label{thm:main_upperboundcol} 
	Let $G$ be a graph with $n$ vertices and cutwidth $\cw$. Given an integer $q\geq 3$ and a prime $p$ that divides $q-1$, there is an $(q-1)^{\cw}n^{O(1)}$ algorithm for counting list $q$-colorings modulo $p$. 
\end{theorem}

\subsection{Definitions and overview}
We first introduce some additional notation and definitions needed in this section.  Let $q$ be an integer and let $p$ be a prime that divides $q-1$.  We are given a graph $G=(V,E)$ with the cutwidth ordering $v_1,\dots,v_n$ of the vertices.
Without loss of generality, we may assume that $G$ is connected. 
We write $G_i=G[\{v_1,\dots,v_i\}]$ and
\[
L_i=\{v\in V(G_i):vv_j\in E \text{ for some } j>i\}.
\]
Note that by definition of cutwidth, $L_i\subseteq L_{i-1}\cup \{v_i\}$ and  $|L_i|\leq \cw$ for all $i$ (since the number of endpoints of a set of edges is upper bounded by the number of edges in the set).

Let $i\in [n]$ be given and write $X_i=L_i\cup \{v_i\}$ for the set of vertices left of the cut that either have an edge in the cut, or are the rightmost vertex left of the cut. We also define $Y_i=\{v_{i+1},\dots,v_n\}\cap N(X_i)$. Figure \ref{fig:ub1} illustrates this notation.

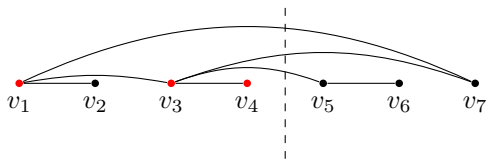
\begin{figure}[h]
	\centering
		\begin{tikzpicture}
		\tikzset{
			EdgeStyle/.append style = {->, bend left} }
		\tikzstyle{vertex}=[circle,fill=black,inner sep=1pt]
		\foreach \x in {1,3,4}
		{
			\node[vertex, label=below:$v_\x$, fill=red] (v\x) at (\x, 0){};
		}
		\foreach \x in {2,5,6,7}
		{
			\node[vertex, label=below:$v_\x$] (v\x) at (\x, 0){};
		}
		\draw (v1) edge[bend left=25] (v7);
		\draw (v3) edge (v4);
		\draw (v3) edge[bend left=20] (v5);
		\draw (v1) edge[bend left=10] (v3);
		\draw[dashed] (4.5, 1)--(4.5, -1);
		\draw (v6) edge (v5);
		\draw (v3) edge[bend left=20] (v7);
		\draw (v2) edge (v1);
	\end{tikzpicture}
	\caption{In the above graph,  $L_4=\{v_1,v_3\}$, $X_4=\{v_1,v_3,v_4\}$ (the red vertices) and $Y_4=\{v_5,v_7\}$.}
	\label{fig:ub1}
\end{figure}

Let $T_i[x]$ be the number of extensions of $x\in \col(X_i)$ to a coloring of $G_i=G[\{v_1,\dots,v_i\}]$. Equivalently, $T_i[x]$ gives the number of colorings of $G_i$ that are compatible with $x$. 

A standard dynamic programming approach builds on the following observation.
\begin{lemma}[Folklore]
	\label{lem:folkloreDP}
	For $x\in \col(X_i)$,
	\[
	T_i[x]=\sum_{\substack{z\in \col(X_{i-1})\\z\sim x}} T_{i-1}[z].
	\]
\end{lemma}
\begin{proof}
    We need to show that for $x\in \col(X_i)$,
	\[
	T_i[x]=\sum_{\substack{z\in \col(X_{i-1})\\z\sim x}} T_{i-1}[z].
	\]
	By definition of the $X_i$, $X_i\setminus X_{i-1}=\{v_i\}$ and so there are no edges between $V(G_{i-1})\setminus X_{i-1}$ and $X_{i}\setminus X_{i-1}$. This means a coloring of $G_{i-1}$ is compatible with $x$ if (and only if) its restriction to $X_{i-1}$ is compatible with $x$. This `cut property' is why the equation holds. 
	
	We now spell out the technical details. Recall that we write $\col(X_i)[x]$ for the set of colorings $z\in \col(X_i)$ that are compatible with $x$. We need to prove that
	\[
	|\col(G_i)[x]|=|\{(z,\phi_{i-1}):z\in \col(X_{i-1})[x],~\phi_{i-1}\in \col(G_{i-1})[z]\}|.
	\]
	Let $\phi\in \col(G_i)[x]$, that is, a list $q$-coloring of $G_i$ compatible with $x$. Then $z=\phi|_{X_{i-1}}$ is compatible with $x$ and $\phi_{i-1}=\phi|_{G_{i-1}}$ is compatible with $z$. So $\phi \mapsto (z,\phi_{i-1})$ maps from the set displayed on the left-hand side to the set displayed on the right-hand side. Its inverse is given by $(z,\phi_{i-1})\mapsto x\cup \phi_{i-1}$. This is well-defined: $x$ and $\phi_{i-1}$ are compatible because $z\in \col(X_{i-1})$ is compatible with $x$, and $\phi_{i-1}|_{X_{i-1}}=z$ (here we use the above mentioned `cut property').
\end{proof}
Since $|\col(X_i)|$ may be of size $q^{|X_i|}$, we cannot compute $T_i$ in its entirety within the claimed time bound. 
The idea of our algorithm is to use the same dynamic programming iteration, but to compute the values of $T_i$ only for a subset $\S_i'\subseteq \col(X_i)$ of the possible colorings which is of significantly smaller size. In fact, we will compute a function $T_i':\S_i'\to \F_p$ that does not necessarily agree with $T_i$ on $\S_i'$. The important property that we aim to maintain is that $T_i'$ carries the `same information' about the number of colorings modulo $p$ as $T_i$ does. This is formalised below. 

\begin{definition} 
Let $H=(X\cup Y, E)$ be a bipartite graph with color compatibility matrix $M$.
Let $T,T':\col(X)\to \F_p$. We say $T'$ is an \emph{$M$-representative} of $T$ if
\[
\sum_{x\in \col(X)} M[x,y]T[x]\equiv_p \sum_{x\in \col(X)}M[x,y]T'[x]  \text{ for all } y\in \col(Y).
\]
\end{definition}
In other words, $T'$ is an $M$-representative of $T$ if $M^\top\cdot T\equiv_p M^\top\cdot T'$.

Above we left the lists and the integer $q$ implicit. We recall that the color compatibility matrix has entries $M[x,y]=1$ if $x\in \col(X)$ and $y\in \col(Y)$ are compatible, and $M[x,y]=0$ otherwise. Let $i\in [n-1]$ be given. Let $M_i$ be the color compatibility matrix of the bipartite graph given by the edges between $X_i$ and $Y_i$. 

Then for $y\in \col(Y_i)$,
\[
\sum_{x\in \col(X_i)}M_i[x,y] T_i[x]
\]
gives the number of colorings of $G_i$ compatible with $y$. 
If we can compute $T_{n-1}'$ that is an $M_{n-1}$-representative of $T_{n-1}$, then by Lemma~\ref{lem:folkloreDP} we can compute the number of $q$-list colorings of the graph (modulo $p$) as 
\[
\sum_{y\in \col(G[v_{n}])} \sum_{x\in \supp(T_{n-1}')}M_{n-1}[x,y] T_{n-1}'[x].
\]
It is an exercise in linear algebra to show that there always exists a $T'$ that $M$-represents $T$ with $|\supp(T')|\leq \rank(M)$. We also need to make sure that we can actually compute this $T'$ within the desired time complexity and therefore reduce the support in a slightly more complicated fashion in Section \ref{subsec:reduce}. We then prove an analogue of Lemma \ref{lem:folkloreDP} in Section \ref{subsec:iterate}, and describe our final algorithm in Section \ref{subsec:alg}.

\subsection{Computing a reduced representative}
\label{subsec:reduce}
In this subsection, we show how to find a function $T'$ that $M$-represents $T$, while decreasing an upper bound on the size of the support of the function.
\begin{definition}
	For a function $f: \col(X) \rightarrow \F_p$ we say that $r\in X$ is a \emph{reduced vertex} if $f(c)=0$ whenever $c(r)=q$. 
\end{definition}

The link between reduced vertices and the support of $T:\col(X)\to \F_p$ is explained as follows.
If $R$ is a set of reduced vertices of $T$, then we can compute a set of colorings containing the support of $T$ of size at most $(q-1)^{|R|}q^{|X|-|R|}$. Indeed, we may restrict to the colorings that do not assign the color $q$ to any vertex in $R$. 

The following result allows us to turn vertices of degree $1$ in $H$ into reduced vertices. The assumption that the vertex has degree $1$ will be useful in proving the result because it implies the associated compatibility matrix can be written as a Kronecker product with $J_q$ and another matrix.
\begin{lemma}\label{lem:repset}
	There is an algorithm \textbf{Reduce} that, given a bipartite graph $H$ with parts $X, Y$ and associated color compatibility matrix $M$,  a function $T: \col(X) \rightarrow \F_p$ with reduced vertices $R \subseteq X$ and a vertex $v \in X \setminus R$ of degree $1$ in $H$, outputs a function $T':\col(X)\to \F_p$ with reduced vertices $R \cup \{v\}$ that is an $M$-representative of $T$. The run time is in $O((q-1)^{|R|}q^{|X| - |R|})$.
\end{lemma}
\begin{proof}
Let $H$ a bipartite graph with parts $X, Y$ and associated color compatibility matrix $M$. Let $T: \col(X) \rightarrow \F_p$ be a function with reduced vertices $R \subseteq X$ and let vertex $v \in X \setminus R$ be a vertex of degree $1$ in $H$. We need to find a function $T':\col(X)\to \F_p$ with reduced vertices $R \cup \{v\}$ that is $M$-representative of $T$. (And need to show this can be done in time $O((q-1)^{|R|}q^{|X| - |R|})$.)

We may restrict to colorings $x$ that do not assign value $q$ to any element of $R$. There are at most $(q-1)^{|R|}q^{|X\setminus R|}$ such colorings. 
	We set 
	\[
	T'[x]=
	\begin{cases}
	0, & \text{if } x(v)=q, \\
	T[x]-T[x'], \text{ where $x'$ is obtained from $x$ by changing the value of $v$ to $q$}, & \text{otherwise}.
	\end{cases}
	\]
	This computation is done in time linear in the number of the colorings $x$ we consider, so the running time is in $O((q-1)^{|R|}q^{|X\setminus R|})$.
	
	First we will show that $T'$ is an $M$-representative of $T$.
	Let $y\in \col(Y)$ be a coloring of the right hand side of the bipartite graph $H$. We need to show that
	\[
	\sum_{\substack{x\in \col(X) \\ x\sim y}} T[x] \equiv_p
	\sum_{\substack{x\in \col(X) \\ x\sim y}} T'[x].
	\]
By definition,
	\[
	\sum_{\substack{x\in \col(X) \\ x\sim y}} T'[x] = 
	\sum_{\substack{x\in \col(X) \\ x\sim y \\ x(v)=q}} 0+
	\sum_{\substack{x\in \col(X) \\ x\sim y\\ x(v) \neq q}} T[x] -  T[x'].
	\]
	Thus it remains to show that 
	\[
	\sum_{\substack{x\in \col(X) \\ x\sim y\\ x(v) \neq q}} -T[x'] \equiv_p \sum_{\substack{x\in \col(X) \\ x\sim y\\ x(v) = q}} T[x].
	\] 
	Let $x\in \col(X)$ with $x(v)= q$.
	We show the equality by proving that the number of times $T[x]$ appears on the left hand side equals the number of times $T[x]$ appears on the right hand side, modulo $p$.
	Let $w\in Y$ be the unique neighbor of the vertex $v$. 
	
	First assume that $x\sim y$. Then $y(w)\neq q$. If we adjust $x$ to the coloring $x_i$, which is equal to $x$ apart from assigning color $i$ to $v$ instead of $q$, then $x_i\sim y$ if and only if $i\neq y(w)$.
Hence the term $-T[x]$ appears $q-2$ times on the left hand side, and $T[x]$ appears once on the right hand side. Since $p$ divides $q-1$, we find $q-2\equiv_p -1$ and hence both contributions are equal modulo $p$.
	
If $x\not\sim y$, then either $x$ does not appear on both sides (because $x|_{X\setminus\{v\}}$ is already incompatible with $y$) or $y(w)=q$. If $y(w)=q$, then the term $T[x]$ appears $q-1\equiv_p 0$ times on the left hand side by a similar argument as the above, and does not appear on the right hand side. This shows the claimed equality and finishes the proof.
\end{proof}
We say that a function $T:\col(X)\to \F_p$ is \emph{fully reduced} if every vertex $v\in X$ of degree 1 is a reduced vertex of $T$.
In order to keep the running time low, we will ensure that $R$ is relatively large whenever we apply Lemma \ref{lem:repset}.
\subsection{Computing $T_i'$ from $T_{i-1}'$}
\label{subsec:iterate}
Recall that $T_i[x]$ gives the number of colorings of $G_i$ that are compatible with $x\in \col(X_i)$ and that $M_i$ is the color compatibility matrix of the bipartite graph between $X_i$ and $Y_i$ (corresponding to the $i$th cut). 
\begin{lemma}\label{lem:iter}
	Let $i\in [n-1]$.
	Suppose that $T'_{i-1}$ is an $M_{i-1}$-representative of $T_{i-1}$ and that $T_{i-1}'$ is fully reduced. Given $T_{i-1}'$ and a set $R_{i-1}$ of reduced vertices for $T_{i-1}'$, we can compute a function $T_i'$ that is an  $M_{i}$-representative of $T_i$ in time $O((q-1)^{|R_{i-1}|}q^{|X_{i-1}|-|R_{i-1}|+1})$, along with a set $R_i$ of reduced vertices for $T_i'$ such that $|X_{i}\setminus R_i|\leq (\cw -|R_i|)/2+1$.
\end{lemma}
\begin{proof}
	Let $i\in [n-1]$ and let $T'_{i-1}$ be $M_{i-1}$-representative of $T_{i-1}$ and fully reduced, with $R_{i-1}$ a set of reduced vertices for $T_{i-1}'$. 
	We need to compute (in time $O((q-1)^{|R_{i-1}|}q^{|X_{i-1}|-|R_{i-1}|+1})$) a function $T_i'$ that is $M_{i}$-representative of $T_i$, along with a set $R_i$ of reduced vertices for $T_i'$, such that $|X_{i}\setminus R_i|\leq (\cw -|R_i|)/2+1$.
	
	We will work over $\mathbb{F}_p$ during this proof, in particular abbreviating $\equiv_p$ to $=$.
	Analogous to Lemma \ref{lem:folkloreDP}, we define, for $x\in \col(X_i)$,
	\begin{equation}
		\label{eq:Tprime}
		T'_{i}[x] = \sum_{\substack{z \in \col(X_{i-1})\\z\sim x}} T'_{i-1}[z].
	\end{equation}
	Note that 
	\[
	\sum_{\substack{z \in \col(X_{i-1})\\z\sim x}} T'_{i-1}[z]=\sum_{\substack{z \in \supp(T_{i-1}')\\z\sim x}} T'_{i-1}[z].
	\]
	We compute $T'_i$ from $T'_{i-1}$ as follows. Let 
	\[
	\S_{i-1}'=\{c\in \col(X_{i-1}):c(r)\neq q \text{ for all }r\in R_{i-1}\}.
	\]
	By the definition of reduced vertex, $\S_{i-1}'$ contains the support of $T_{i-1}'$ since $T'_{i-1}$ is fully reduced. 
	Recall that $X_i\setminus \{v_i\} \subseteq X_{i-1}$, so any $x\in \col(X_i)$ is determined if we provide colors for the vertices in $X_{i-1}\cup \{v_i\}$.
	For a color $c\in [q]$, let $f_c:\{v_i\}\to \{c\}$ be the function that assigns color $c$ to $v_i$. 
	For each $z\in \S_{i-1}'$, for each $c\in [q]$ for which $z\sim f_c$, we compute
	\[
	x=(z\cup f_c)|_{X_i}\in \col(X_i)
	\]
	and increase $T_i'[x]$ by $T_{i-1}'[z]$ if it has been defined already, and initialise it to $T_{i-1}'[z]$ otherwise. The remaining values are implicitly defined to 0. The running time is as claimed because $|\S_{i-1}'|\leq (q-1)^{|R_{i-1}|}q^{|X_{i-1}|-|R_{i-1}|}$ and $|[q]|\leq q$.
	
	Next, we compute a set $R_i$ of reduced vertices for $T_i'$. 
	We set $R_i=X_i\setminus (A_i\cup B_i\cup \{v_i\})$, where
	\[
	A_i=\{u\in X_i\setminus \{v_i\}\: : \: |N(u)\cap Y_i| \geq 2\}
	\]
	and
	\[
	B_i=\{u\in X_i \setminus \{v_i\}\: :\: |N(u)\cap Y_i|=1  \text{ and } uv_i\in E \}.
	\]
	It is easy to see that $A_i$ and $B_i$ are disjoint. Within the $(i-1)$th cut, each vertex in $A_i\cup B_i$ has at least two edges going across the cut, so $|R_i|+2|A_i|+2|B_i|\leq \cw$. Therefore, $|X_i\setminus R_i|\leq (\cw-|R_i|)/2+1$. 
	
	We now show that $R_i$ is indeed a set of reduced vertices. Suppose not, and let $r \in R_i$ and $c\in \col(X_i)$ with $c(r)=q$ yet $T_{i}'[c]\neq 0$. Since $T_i'[c]\neq 0$, there exists $z\in \col(X_{i-1})$ with $z\sim c$ and $T_{i-1}'[z]\neq 0$. By definition $r\in X_i\setminus\{ v_i\}\subseteq X_{i-1}$. Moreover, $z(r)=q$ since $z\sim c$ and $c(r)=q$. Therefore $r$ is not reduced for $T_{i-1}'$. We now show $r$ moreover has degree $1$ in the bipartite graph between $X_{i-1}$ and $Y_{i-1}$ (corresponding to the $(i-1)$th cut), contradicting our assumption that  $T_{i-1}'$ is fully reduced. Since $r\not\in A_i\cup B_i$, it has at most one edge going over the $(i-1)$th cut. Moreover, $r\in X_{i}\setminus \{v_i\}\subseteq L_{i}$, and so it has at least one edge to $Y_i\subseteq Y_{i-1}$. So $r$ has exactly one neighbor in $Y_{i-1}$.
	
	It remains to prove that $T'_i$ is $M_i$-representative of $T_i$. 
 All equalities and computations below take place modulo $p$. Let $y\in \col(Y_i)$. We need to show that
	\begin{equation}
		\label{eq:represents}
		\sum_{x\in \col(X_i)}  M_i[x,y] T_i[x]= \sum_{x\in \col(X_i)} M_i[x,y] T'_i[x].
	\end{equation}
	By Lemma \ref{lem:folkloreDP}, for all $x\in \col(X_i)$,
	\[
		T_{i}[x] = \sum_{z \in \col(X_{i-1})[x]} T_{i-1}[z],
	\]
and we crucially use in the computation below that we used the same expression when definining $T_i'$, which will allow us to exploit the fact that $T_{i-1}'$ is an $M_{i-1}$-representative of $T_{i-1}$.
We find
	\begin{align*}
		\sum_{x\in \col(X_i)}  M_i[x,y]T_i[x]
		&=  \sum_{x\in \col(X_i)}  \left(\sum_{z\in \col(X_{i-1})[x]} T_{i-1}[z] \right) M_i[x,y]\\ 
		&= \sum_{z\in \col(X_{i-1})}  \sum_{\substack{x\in \col(X_i)\\x\sim z\\ x\sim y}} T_{i-1}[z].
		\intertext{In the second sum, recall that each $x\in \col(X_i)$ compatible with $z$ is of the form $x=(z\cup f_c)|_{X_i}$ with $f_c,z$ compatible and $f_c:\{v_i\}\to \{c\}$ for some $c\in [q]$. We find that $x$ is compatible with $y$ if and only if $f_c$ is compatible with $y$ and $z$ is compatible with $y$, so the previous expression equals }
		&= \sum_{z\in \col(X_{i-1})}  \sum_{f_c\in \col(\{v_i\})} T_{i-1}[z]
		1_{z\sim f_c} 1_{z\sim y} 
		1_{f_c,\sim y }\\
		&= \sum_{f_c\in \col(\{v_i\})}
		1_{f_c\sim y}
		\sum_{z\in \col(X_{i-1})} 
		T_{i-1}[z]
		1_{z\sim f_c} 
		1_{z\sim y},\\
		\intertext{where $1_{b}$ denotes the indicator variable which is equal to $1$ if condition $b$ is true and equals $0$ otherwise. Since $Y_{i-1}\subseteq Y_i\cup \{v_i\}$, we can consider $y'=(f_c\cup y)|_{Y_{i-1}}$ and find $M_{i-1}[z,y']=1$ exactly when $z$ is compatible with both $f_c$ and $y$. Using this to rewrite the previous expression, and then noting that $y'\in \col(Y_{i-1})$ and $T_{i-1}'$ is a $M_{i-1}$-representative of $T_{i-1}$, we find}
		\sum_{x\in \col(X_i)}  M_i[x,y]T_i[x] &=\sum_{f_c\in \col(\{v_i\})}
		1_{f_c\sim y}
		\sum_{z\in \col(X_{i-1})} 
		T_{i-1}[z] M_{i-1}[z,(f_c\cup y)|_{Y_{i-1}}]\\
		&= \sum_{f_c\in \col(\{v_i\})}
		1_{f_c\sim y}
		\sum_{z\in \col(X_{i-1})} 
		T'_{i-1}[z] M_{i-1}[z,(f_c\cup y)|_{Y_{i-1}}],\\
		\intertext{which for the same reasons as above }
		&= \sum_{f_c\in \col(\{v_i\})}
		1_{f_c\sim y}
		\sum_{z\in \col(X_{i-1})} 
		T'_{i-1}[z]1_{z\sim f_c} 
		1_{z\sim y}\\
		&= \sum_{z\in \col(X_{i-1})}  \sum_{f_c\in \col(\{v_i\})} T'_{i-1}[z]
		1_{z\sim f_c} 1_{z\sim y} 
		1_{f_c\sim y}\\
		&= \sum_{z\in \col(X_{i-1})}  \sum_{\substack{x\in \col(X_{i})\\x\sim z\\ x\sim y}} T'_{i-1}[z]\\
		&=  \sum_{x\in \col(X_i)}  \left(\sum_{z\in \col(X_{i-1})[x]} T'_{i-1}[z] \right) M_i[x,y]\\ 
 		&=\sum_{x\in \col(X_{i})}  M_i[x,y]T'_i[x].
	\end{align*}
	This shows (\ref{eq:represents}) and completes the proof.
\end{proof}

\subsection{Analysis of final algorithm}
\label{subsec:alg}
We initialize $T_1 = \textbf{1}$, the all-ones vector. Indeed, each $x\in \col(\{v_1\})$ has a unique extension to $G_1$ (namely itself). We then repeatedly apply the \textbf{Reduce} algorithm from Lemma \ref{lem:repset} until we obtain a fully reduced function $T_1'$ that is an $M_1$-representative of $T_1$, with some set of reduced vertices $R_1$. For $i=2,\dots,n$, we repeat the following two steps.
\begin{enumerate}
	\item Apply Lemma \ref{lem:iter} with inputs $(T_{i-1}',R_{i-1})$ in order to obtain the vector $T'_{i}$ that is an $M_i$-representative of $T_i$, and a set of reduced vertices $R_i$ for $T'_i$. 
	\item While $X_i\setminus R_i$  has a vertex $v$ of degree 1, apply the \textbf{Reduce} algorithm from Lemma \ref{lem:repset} to $(T'_{i},R_{i})$, and add $v$ to $R_i$. 
\end{enumerate}
At the end of step 2,  we obtain a fully reduced function $T_i'$ that is an $M_i$-representative of $T_i$. Moreover, the set $R_i$ of reduced vertices has only increased in size compared to the set we obtained in step 1. We apply Lemma \ref{lem:repset} at most $|X_i|$ times in the second step.

We eventually compute $T'_{n-1}$ that is an $M_{n-1}$-representative of $T_{n-1}$ with a fully reduced set $R_{n-1}$. We output
\[
\sum_{y\in \col(Y_{n-1})}   \sum_{x\in \col(X_{n-1})} T'_{n-1}[x]M_{n-1}[x,y].
\]
Since $T_{n-1}'$ is an $M_{n-1}$-representative of $T_{n-1}$, this gives the number of list colorings of $G$ modulo $p$. We may compute the expression above efficiently by  reducing the second summation to the colorings in 
\[
\S'_{n-1}=\{c\in \col(X_{n-1}):c(r)\neq q \text{ for all }r\in R_{n-1}\}.
\]

The total running time is now bounded by
\[
C \sum_{i=1}^{n-1}|X_i|(q-1)^{|R_i|}q^{|X_i|-|R_i|}
\]
for some constant $C>0$. By Lemma \ref{lem:iter}, $|X_i\setminus R_i|\leq  (\cw-|R_i|)/2+1$ for all $i\in [n-1]$. For $q\geq 3$, $q^{1/2}<q-1$ and so 
\[
(q-1)^{|R_i|}q^{|X_i|-|R_i|}\leq q(q-1)^{|R_i|}(q^{1/2})^{\cw-|R_i|} < q(q-1)^{\cw}.
\]
This shows the total running time is of order $(q-1)^{\cw}n^{O(1)}$. This finishes the proof of Theorem \ref{thm:main_upperboundcol}.

\section{Lower bounds}
\label{sec:collow}
There exists an efficient reduction from \textsc{SAT} to the problem $\#_p$\textsc{SAT} of counting the number of satisfying assignments for a given boolean formula modulo $p$ \cite{calabro2008complexity}. There also exists a reduction from \textsc{SAT} to \textsc{CSP}$(q, r)$, which preserves the number of solutions \cite{focke2022counting}. Putting these two together gives a reduction from \textsc{SAT} to $\#_p$\textsc{CSP}$(q, r)$.

In this section we give a reduction from  $\#_p$\textsc{SAT} to  $\#_p$\textsc{List $q$-Coloring}, the problem of counting the number of valid list $q$-colorings of a given graph $G$ with color lists $(L_v)_{v \in V(G)}$. We use this to conclude the lower bounds of Theorem \ref{thm:main} and Theorem \ref{thm:main_lbcse}.

\subsection{Controlling the number of extensions modulo $p$}
Our main gadget can be attached to a given set of vertices, and has the property that for each precoloring of the `glued on' vertices, there is a specified number of extensions. This is made precise in the result below.

\begin{theorem} \label{thm:ColGad}
Let $k \in \mathbb{N}$ and $f : [q]^k \to \mathbb{N}$. There exist a graph $G_f$, a set of vertices $B = \{b_1, \dots, b_k\} \subseteq V(G_f)$ of size $k$ and lists $(L_v)_{v\in V(G_f)}$, such that for any  $\alpha \in [q]^k$, there are exactly $f(\alpha)$ list $q$-colorings $c$ of $G_f$ with $c(b_i) = \alpha(i)$ for all $i\in [k]$. Additionally,  $|V(G_f)| \leq20kq^{k+1}\max(f)$ 
and $G_f$ has cutwidth at most $6kq^{k+2}$.
\end{theorem}
The proof is given in Appendix \ref{sec:collowomit}.

\subsection{Reduction for counting $q$-colorings modulo $p$}
In this section we prove the following result.
\begin{theorem}
\label{thm:ColRed}
Let $p$ be a prime and let $q \in \mathbb{N}$ such that $p$ does not divide $q-1$. Assuming SETH, there is no $\varepsilon>0$ for which there exists an algorithm that counts the number of list $q$-colorings modulo $p$ for a given $n$-vertex graph, with a given cut decomposition of width $\cw$, in time $(q-\varepsilon)^\cw n^{O(1)}$.
\end{theorem}

Suppose now that $p$ divides $q-1$. Let $q'=q-1$. Then $p$ does not divide $q'-1 = q-2$ and so the result above applies. Noting that any algorithm for \#\textsc{List $q$-Coloring} also works for \#\textsc{List $q'$-Coloring}, we find the following corollary.
\begin{corollary}
\label{cor:ColRed}
Let $p$ be a prime and let $q \in \mathbb{N}$ such that $p$ divides $q-1$. Assuming SETH, there is no $\varepsilon>0$ for which there exists an algorithm that counts the number of list $q$-colorings modulo $p$ for a given $n$-vertex graph, with a given cut decomposition of width $\cw$, in time $(q-1-\varepsilon)^\cw n^{O(1)}$.  
\end{corollary}
Combining the two results above with Theorem \ref{thm:main_upperboundcol} gives Theorem \ref{thm:main}. 

We use the notion of \textit{constraint satisfaction problems} ($\textsc{CSP}$). Informally, a $\textsc{CSP}$ asks if there is an assignment of values from a given domain to a set of variables such that they satisfy a given set of relations.
We denote by $\textsc{CSP}(q,r)$ the $\textsc{CSP}$ with domain $[q]$ and constraints of arity at most $r$. We use \#\textsc{CSP}$(q,r)$ to denote the problem of counting the number of solutions of a given instance of $\textsc{CSP}(q,r)$.
We use the following result from \cite{focke2022counting}.

\begin{theorem}[\cite{focke2022counting}, Theorem 2.5]
\label{thm:csphard}
For each prime $p$, for every integer $q \geq 2$ and $\eps > 0$ there is an integer $r$, such that the following holds.  Unless the SETH fails, $\#_p$\textsc{CSP}$(q, r)$ with $n$ variables and $m$ constraints cannot be solved in time $(q -\eps)^n (n + m)^{O(1)}$.
\end{theorem}
This theorem follows from the proof of \cite[Theorem 2.5]{focke2022counting}, since their reduction preserves the number of solutions.

\begin{proof}[Proof of Theorem \ref{thm:ColRed}] 
Let $q \in \mathbb{N}$ and let $p$ be a prime that does not divide $q-1$. Fix $\epsilon>0$ and let $r$ be given from Theorem \ref{thm:csphard}. We will reduce a given instance of $\#_p$\textsc{CSP}$(q, r)$ with constraints $C_1,\dots,C_m$ and variables $x_1,\dots,x_n$ to an instance $(G,L)$ of $\#_p$\textsc{List $q$-Coloring} on $O_{p,r,q}(nm)$ vertices of cutwidth $n+O_{p,r,q}(1)$.

The graph $G$ contains $2m$ columns with $n$ vertices: for each constraint $C_j$, and for each variable $x_i$, we create two vertices $s_{i,j}$ and $t_{i,j}$ (where $j\in [m]$ and $i\in [n]$), which all get $\{1,\dots,q\}$ as list. For all $j\in [m-1]$, we place an edge between $t_{i,j}$ and $s_{i,j+1}$.

The color assigned to $s_{i,1}$ will be interpreted as the value assigned to variable $x_i$. Fix $j\in [m]$. We create gadgets on some vertex set $V_j$ using Theorem \ref{thm:ColGad}, that are `glued' on subsets of vertices from $C_j=\{s_{i,j},t_{i,j}:i\in [n]\}$. 
\begin{enumerate}
\item For each $i\in [n]$, if $j<m$, we create a gadget on boundary set $\{s_{i,j},t_{i,j}\}$ which ensures that we may restrict to counting list colorings $c$ of $(G,L)$ with $c(s_{i,j})=c(s_{i,j+1})$.
\item There is a gadget on a boundary set of size at most $r$ (the $s_{i,j}$ corresponding to the variables involved in the $j$th constraint), for which the number of extensions of any coloring of the boundary to this gadget is equivalent to $0$ modulo $p$  whenever the $j$th constraint is not satisfied, and equal to one otherwise.
\end{enumerate}
A broad overview of the construction is depicted in Figure \ref{fig:Col_Con}.
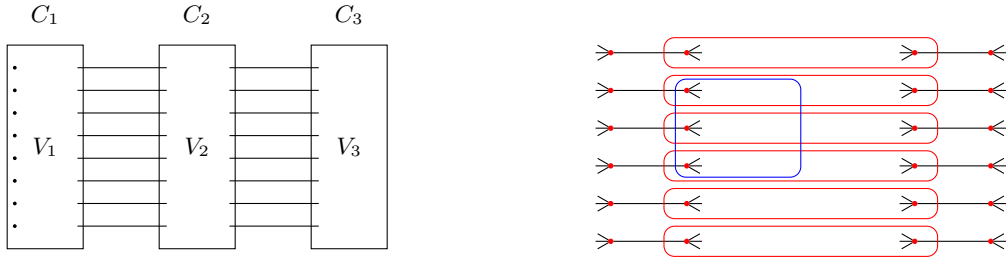
\begin{figure}[htb]
	\centering
	\begin{tabular}{p{0.4\textwidth}p{0.1\textwidth}p{0.4\textwidth}}
		\begin{tikzpicture}
	\tikzstyle{redv} = [draw, circle, fill=red, color=red, inner sep=0.2mm]
	\tikzstyle{smallv} = [draw, circle, fill=black, color=black, inner sep=0.05mm]
	\tikzstyle{blackv} = [draw, circle, fill=black, color=black, inner sep=0.2mm]
	\tikzstyle{tmpv} = [draw, circle, fill=white, color=white, inner sep=0.05mm]
	\def \Rws{8}
	\def \RwHgt{0.3}
	\def \BlkWdth{1}
	\def \Blks{3}
	\pgfmathsetmacro{\Blksm}{\Blks-1}
	\pgfmathsetmacro{\Blksmm}{\Blks-2}
	\foreach \x in {0,...,\Blksm}
	{
		\node (bl\x) at (\x*2, 0){};
		\node (tr\x) at (\x*2+\BlkWdth, \Rws*\RwHgt + \RwHgt){};
		\draw (bl\x) rectangle (tr\x); 
	};

	\foreach \y in {1, ..., \Rws}
	{
		\node[smallv] at (0.1, \y*\RwHgt){};
	};

\foreach \x in {0,...,\Blksmm}
{
	\foreach \y in {1, ..., \Rws}
	{
		\node[tmpv](rvert1\x) at (\x*2 +\BlkWdth - 0.1, \y*\RwHgt) {};
		\draw (rvert1\x)--(\x*2+2 + 0.1, \y*\RwHgt);
	};
};
	\foreach \x in {1,...,\Blks}
	{
		\node at (\x*2 - 2 + 0.5*\BlkWdth, \Rws*\RwHgt + \RwHgt + 0.4){$C_\x$};
		\node at (\x*2 - 2 + 0.5*\BlkWdth, \Rws*\RwHgt*0.5 + 0.5*\RwHgt){$V_\x$};
	}:
	\end{tikzpicture} 
& & \begin{tikzpicture}
	\tikzstyle{redv} = [draw, circle, fill=red, color=red, inner sep=0.2mm]
	\tikzstyle{smallv} = [draw, circle, fill=black, color=black, inner sep=0.05mm]
	\tikzstyle{blackv} = [draw, circle, fill=black, color=black, inner sep=0.2mm]
	\def \RwsGd{6}
		
		\foreach \x in {1,...,\RwsGd}
		{
			\node[redv] (a\x) at (0, \x/2){};
			\node[redv] (b\x) at (1, \x/2){};
			\node[redv] (c\x) at (4, \x/2){};
			\node[redv] (d\x) at (5, \x/2){};
			\draw(a\x)--(b\x);
			\draw(c\x)--(d\x);

		};
	
	\foreach \x in {1,...,6}
	{
		\draw[rounded corners, draw=red] (0.7,\x/2-0.2) rectangle (4.3, \x/2+0.2);
	};
	\draw[rounded corners, draw=blue] (0.85,1.35) rectangle (2.5, 2.65);
	
	\foreach \x in {a,c}
	{
		\foreach \y in {1,...,\RwsGd}
		{
			\draw(\x\y)--++(150:2mm);
			\draw(\x\y)--++(180:2mm);
			\draw(\x\y)--++(210:2mm);
		};
	};
	\foreach \x in {b,d}
	{
		\foreach \y in {1,...,\RwsGd}
		{
			\draw(\x\y)--++(30:2mm);
			\draw(\x\y)--++(0:2mm);
			\draw(\x\y)--++(330:2mm);
		};
	};

\end{tikzpicture}
	\end{tabular}
	
	\caption{A sketch overview of the construction is given on the left-hand side and a more detailed view of two of the columns is given on the right-hand side. The red areas ensure the preservation of information, as described in point 1. The blue area checks whether the clause is satisfied, as described in point 2. }
	\label{fig:Col_Con}
\end{figure}

For the first property, we need the fact that $p$ does not divide $q-1$: this ensures that the color compatibility matrix of a single edge is invertible, which will allow us to `transfer all information about the colors'. The precise construction of the gadgets is deferred to Appendix \ref{sec:collowomit}.

We obtain a cutwidth decomposition of the graph by first running over the vertices in the order 
\[
C_1\cup V_1, ~C_2\cup V_2,\dots,C_m\cup V_m.
\]
Within $C_j\cup V_j$, we first list $s_{1,j},t_{1,j}$ and the vertices in the gadget that has those vertices as boundary set, and then repeat this for $s_{2,j},t_{2,j}$, etcetera. Finally, we run over the vertices in the gadget that verifies the $j$th constraint. At each point, the cutwidth is bounded by $n$ plus a constant (that may depend on $p$, $q$ and $r$).
\end{proof}

\subsection{Corollaries}
We now extend the lower bound of Theorem \ref{thm:ColRed} to counting connected edge sets via the following problem.
\begin{definition}
Given a graph $G$, two $q$-colorings $c$ and $c'$ are equivalent if there is some permutation $\pi: [q] \to [q]$ such that $c = \pi \circ c'$. We will refer to these equivalence classes as \emph{essentially distinct $q$-colorings} and denote the problem of counting the number of essentially distinct $q$-colorings modulo a prime $p$ by $\#_p$\textsc{Essentially distinct $q$-coloring}.
\end{definition}

A simple reduction now gives us the following lower bound for $\#_p$\textsc{Essentially distinct $q$-coloring}. 
\begin{corollary} \label{cor:DisColLow}
Let $p$ be a prime and $q \in \mathbb{N}$ an integer such that $p$ does not divide $q-1$. Assuming SETH, there is no $\epsilon > 0$ for which there exists an algorithm that counts the number of essentially distinct $q$-colorings mod $p$ for a given $n$-vertex graph that is not $(q-1)$-colorable, with a given cut decomposition of cutwidth $\ctw$, in time $(q - \epsilon)^\ctw n^{O(1)}$. 
\end{corollary}
\begin{proof}
Let $(G,L)$ be an instance of list coloring with cut decomposition $v_1,\dots,v_n$. We construct an instance of  $\#_p$\textsc{Essentially distinct $q$-coloring}. The graph $G'$ has vertex set
\[
V(G') = V(G) \cup \{u_c^i:c\in [q],i\in [n]\}.
\]
We add edges such that the vertices $\{u_c^i:c\in [q]\}$ induce a $q$-clique for all $i\in [n]$, and for $i\in [n-1]$ we add the edges $u_c^iu_{c'}^{i+1}$ for all $c\neq c'$. This ensures that, if $u_c^1$ is colored $c$, then $u_c^i$ is colored $c$ for all $i\in [n]$. We now also add edges $u_c^iu_i$ for all $c\not\in L_{v_i}$.
\begin{figure}[htb]
	\centering
	\begin{tikzpicture}
	\tikzstyle{redv} = [draw, circle, fill=red, color=red, inner sep=0.2mm]
	\tikzstyle{smallv} = [draw, circle, fill=black, color=black, inner sep=0.05mm]
	\tikzstyle{blackv} = [draw, circle, fill=black, color=black, inner sep=0.2mm]
		\node (clip1) at (-2, -1){};
		\node (clip2) at (8, 3){};
		\clip (clip1) rectangle (clip2);
		
		\node	(G)		at (-1,0)	{$G$};
		\node	(Gp)	at (-1,1.5)	{$G'\setminus G$};
		
		\foreach \x in {1,2,3}{
		\node[blackv, label = {[shift={(0,-0.8)}]$v_\x$}]		(v\x)	at (3*\x-2,0)	{};
		\node[redv, label = {[shift = {(0.4,-0.6)}]$u^\x_1$}]	(u\x1)	at (3*\x-2,1)	{};
		\node[redv, label = {[shift = {(0,0.2)}]$u^\x_2$}]		(u\x2)	at (3*\x-2.5,2)	{};
		\node[redv, label = {[shift = {(0,0.2)}]$u^\x_3$}]		(u\x3)	at (3*\x-1.5,2)	{};
		\draw[red] (u\x1) -- (u\x2);
		\draw[red] (u\x1) -- (u\x3);
		\draw[red] (u\x2) -- (u\x3);
		};
		
		\foreach[evaluate=\x as \y using int(\x+1)] \x in {1,2}{
		\draw (u\x1) -- (u\y2);
		\draw (u\x1) -- (u\y3);
		\draw (u\x2) -- (u\y1);
		\draw (u\x2) to[out=10, in=170] (u\y3);
		\draw (u\x3) -- (u\y2);
		\draw (u\x3) -- (u\y1);
		};
	
		\draw (v1) -- (v2);
		\draw (v1) --++ (225:0.5);
		\draw (v2) --++ (315:0.5);
		\draw (v2) --++ (225:0.5);
		\draw (v3) --++ (315:0.5);
		\draw (v3) --++ (225:0.5);
		
		\draw (v1) -- (u11);
		\draw (v1) to[out=150,in=225] (u12);
		\draw (v2) -- (u21);
		\draw (v3) -- (u31);
		\draw (v3) to[out=30,in=315] (u33);
		
\end{tikzpicture}
	\caption{Example of the construction on (a part of) a graph $G$, with the cliques indicated in red. In this case $q=3$ and we have $L_{v_1} = \{3\}, L_{v_2} = \{2,3\}$ and $L_{v_3} = \{2\}$.}
\end{figure}
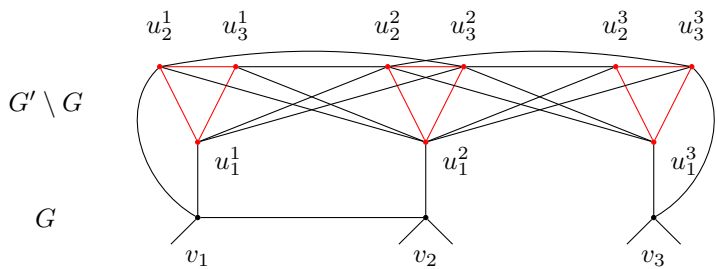
Our new cut decomposition is 
\[
v_1,u_1^1,\dots,u^1_q,v_2,u_1^2,\dots,u^{n-1}_q,v_n,u_1^n,\dots,u_q^n.
\]
Note that $\cw(G')\leq \cw(G)+q^2$, $|V(G')|\leq (q + 1)|V(G)|$ and that $G'$ is not $(q-1)$-colorable.  By Theorem \ref{thm:main}, it suffices to show that the number of essentially distinct $q$-colorings of $G'$ equals the number of list $q$-colorings of $(G,L)$. We will do this by defining a bijective map.

	
Let $\alpha$ be a list coloring of $(G,L)$. Then we can color $G'$ by setting $\alpha'(v) = \alpha(v)$ for $v \in V(G)$ and $\alpha'(u_c^i) = c$ for $c \in [q]$ and $i\in [n]$. This gives us a mapping $\gamma: \alpha \mapsto \overline{\alpha'}$, where $\overline{\alpha'}$ is the equivalence class of $\alpha'$. We find an inverse map by first fixing a representative $\alpha'$ for $\overline{\alpha'}$, such that $\alpha'(u_c^1) = c$ for $c \in [q]$. We can do this since $G'[\{u_1^1, \dots, u_q^1\}]$ is a clique and thus each $u_c^1$ must get a unique color. Also note that since every color is now used, the rest of the coloring is also fixed and thus we find a unique representative this way. We now map $\overline{c'}$ to $c'|_{V(G)}$. Note that these two maps are well defined and compose to the identity map. We conclude that the number of list colorings of $(G,L)$ is equal to the number of essentially distinct colorings of $G'$.
\end{proof}

To achieve the lower bound in Theorem \ref{thm:main_lbcse}, we use an existing argument from \cite{annan1994complexity} to extend this bound to $\#_p$\textsc{Connected Edge Sets}. 
 For this we will need the following definition.
\begin{definition}
	The \emph{$k$-stretch} of a graph $G$ is the graph obtained from $G$ by replacing each edge with a path of length $k$. We denote the $k$-stretch of $G$ by $^kG$.
\end{definition}
Note that $^kG$ has the same cutwidth as $G$.

\begin{theorem}
Assuming SETH, there is no $\epsilon>0$ for which there exists an algorithm that counts the number of spanning connected edge sets mod $p$ of $n$-vertex graphs of cutwidth at most $\ctw$ in time $O((p - \epsilon)^\ctw n^{O(1)})$.    
\end{theorem}
\begin{proof}
	This proof closely follows a reduction from Annan \cite{annan1994complexity}, using ideas from Jaeger, Vertigan and Welsh \cite{jaeger1990computational}.
	
	Let $G$ be any graph with cutwidth $\ctw$ and let $p$ be a prime. Note that the number of spanning connected edgesets of $G$ is equal to the value of $T(G;1,2)$, the Tutte polynomial of $G$ at the point $(1,2)$. The following formula is found in (\cite{jaeger1990computational}, proof of Theorem 2)
	\[
	T(^kG;a,b) = (1 + a + \dots + a^{k-1})^{n - r(G)} T\left(G;a^k,\frac{b + a + \dots + a^{k-1}}{1 + a + \dots + a^{k-1}}\right).
	\]
	Choosing $a=1, b = 2$ and $k = p - 1$, gives
	\[
	T(^{p-1}G;1,2) = (p - 1)^{n - r(G)}T\left(G;1,\frac{2 + p - 2}{p-1}\right) \equiv_p (-1)^{n-r(G)} T(G;1-p,0).
	\]
	Here we use the fact that for any multivariate polynomial $P(x,y) \equiv_p P(x + tp, y + sp)$ for any $s, t \in \mathbb{Z}$. We find that, since the $k$-stretch of a graph $G$ has the same cutwidth as $G$, an algorithm that counts the number of spanning connected edgesets (mod $p$) on a graph with bounded cutwidth, also gives the Tutte polynomial at $(1-p,0)$ mod $p$. 
	
	Using another well known interpretation of the Tutte polynomial \cite{oxley1979tutte} we can relate the $T(G;1-p,0)$ to the chromatic polynomial $P(G;p)$ as follows:
	\[
	P(G;p) = (-1)^{r(G)}p^{k(G)}T(G;1-p,0).
	\]
	Note that we may assume that the number of connected components $k(G) = 1$ (and thus $r(G) = n - 1$), since the number of spanning connected edgesets is trivially $0$ if $G$ has more than one component. We want to get rid of the remaining factor of $p$ (since we will work mod $p$). To do this we will count the number of essentially distinct colorings (different up to color-permutations) using exactly $p$ colors instead. This will turn out to be at least as hard as counting all colorings.
	
	Let $G$ be a graph that is not $(p-1)$-colorable. With this assumption, the number of $p$ colorings of $G$ is $p!$ times the number $C_p(G)$ of essentially distinct $p$-colorings of $G$, since any coloring uses all colors and thus can be mapped to $p!$ equivalent colorings by permuting the colors. So
	\[
	(-1)^{n-1}pT(G,1-p,0) =P(G;p) = p(p-1)!C_p(G),
	\]
	which holds over the real numbers hence we may divide both sides by $p$. By Wilson's Theorem $(p-1)! \equiv_p -1$, so we find 
	\[
	(-1)^{n-1}T(G,1-p,0) \equiv_p -C_p(G).
	\]
	Hence we can use the number of spanning connected edgesets (mod $p$) of the $(p-1)$-stretch of $G$ to find the  Tutte polynomial at $(1-p,0)$ mod $p$ and then also the number of essentially distinct colorings. The claimed result now follows from Corollary \ref{cor:DisColLow} (with $q=p$). 
\end{proof}
The lemma above gives the lower bound of Theorem \ref{thm:main_lbcse} The upper bound is proved in the appendix.

\section{Conclusion}
In this paper we give tight lower and upper bounds for counting the number of (list) $q$-colorings and connected spanning edge sets of graphs with a given cutwidth decomposition of small cutwidth. Our results specifically relate to list $q$-coloring and essentially distinct $q$-coloring, but they can easily be extended to normal $q$-coloring for certain cases. In particular, if $q < p$, we may apply Corollary \ref{cor:DisColLow}, since in the setting of the corollary, the values differ by $q!$ which is nonzero modulo $p$. If
the chromatic number $\chi(G) \geq p$, then the number of $q$-colorings is trivially $0 \mod p$, since the number of $q$-colorings is a multiple of $\chi(G)$. This leaves us with the rather specific case of $\chi(G) < p \leq q$, for which the exact complexity remains unresolved.

Our results on the modular counting of colorings show that the modulus can influence the complexity in interesting ways, and that in some cases this effect can be directly explained by the rank of the compatibility matrix.
 
Our results leave several directions for further research:
\begin{itemize}
	\item What is the fine-grained complexity of evaluating other points of the Tutte polynomial (modulo $p$)?
	\item What is the complexity of counting homomorphisms to graphs different from complete graphs, e.g. cycles or paths. Is it still determined by the rank of an associated compatibility matrix?
\end{itemize}

Another question in the direction of fast \emph{decision problems} is how small representative sets we can get for the compatibility matrix of graphs other than complete bipartite graphs (which is equivalent to the setting of \cite[Theorem 1.2]{FominLPS16}) or matchings which has been studied for the decision version in~\cite{jansen2019computing}.

\bibliography{bib}

\begin{thebibliography}{10}

\bibitem{annan1994complexity}
James~Douglas Annan.
\newblock {\em The complexity of counting problems}.
\newblock PhD thesis, University of Oxford, 1994.

\bibitem{bjorklund2006inclusion}
Andreas Bj{\"o}rklund and Thore Husfeldt.
\newblock Inclusion-exclusion based algorithms for graph colouring.
\newblock In {\em Electronic Colloquium on Computational Complexity (ECCC)},
  volume~13. Citeseer, 2006.

\bibitem{BjorklundHKK08}
Andreas Bj{\"{o}}rklund, Thore Husfeldt, Petteri Kaski, and Mikko Koivisto.
\newblock Computing the tutte polynomial in vertex-exponential time.
\newblock In {\em 49th Annual {IEEE} Symposium on Foundations of Computer
  Science, {FOCS} 2008, October 25-28, 2008, Philadelphia, PA, {USA}}, pages
  677--686. {IEEE} Computer Society, 2008.
\newblock \href {https://doi.org/10.1109/FOCS.2008.40}
  {\path{doi:10.1109/FOCS.2008.40}}.

\bibitem{BodlaenderCKN15}
Hans~L. Bodlaender, Marek Cygan, Stefan Kratsch, and Jesper Nederlof.
\newblock Deterministic single exponential time algorithms for connectivity
  problems parameterized by treewidth.
\newblock {\em Inf. Comput.}, 243:86--111, 2015.
\newblock \href {https://doi.org/10.1016/j.ic.2014.12.008}
  {\path{doi:10.1016/j.ic.2014.12.008}}.

\bibitem{calabro2008complexity}
Chris Calabro, Russell Impagliazzo, Valentine Kabanets, and Ramamohan Paturi.
\newblock The complexity of unique k-sat: An isolation lemma for k-cnfs.
\newblock {\em Journal of Computer and System Sciences}, 74(3):386--393, 2008.

\bibitem{CurticapeanLN18}
Radu Curticapean, Nathan Lindzey, and Jesper Nederlof.
\newblock A tight lower bound for counting hamiltonian cycles via matrix rank.
\newblock In Artur Czumaj, editor, {\em Proceedings of the Twenty-Ninth Annual
  {ACM-SIAM} Symposium on Discrete Algorithms, {SODA} 2018, New Orleans, LA,
  USA, January 7-10, 2018}, pages 1080--1099. {SIAM}, 2018.
\newblock \href {https://doi.org/10.1137/1.9781611975031.70}
  {\path{doi:10.1137/1.9781611975031.70}}.

\bibitem{CurticapeanM16}
Radu Curticapean and D{\'{a}}niel Marx.
\newblock Tight conditional lower bounds for counting perfect matchings on
  graphs of bounded treewidth, cliquewidth, and genus.
\newblock In Robert Krauthgamer, editor, {\em Proceedings of the Twenty-Seventh
  Annual {ACM-SIAM} Symposium on Discrete Algorithms, {SODA} 2016, Arlington,
  VA, USA, January 10-12, 2016}, pages 1650--1669. {SIAM}, 2016.
\newblock \href {https://doi.org/10.1137/1.9781611974331.ch113}
  {\path{doi:10.1137/1.9781611974331.ch113}}.

\bibitem{CyganKN18}
Marek Cygan, Stefan Kratsch, and Jesper Nederlof.
\newblock Fast hamiltonicity checking via bases of perfect matchings.
\newblock {\em J. {ACM}}, 65(3):12:1--12:46, 2018.
\newblock \href {https://doi.org/10.1145/3148227} {\path{doi:10.1145/3148227}}.

\bibitem{cygan2011solving}
Marek Cygan, Jesper Nederlof, Marcin Pilipczuk, Michal Pilipczuk, Joham~MM van
  Rooij, and Jakub~Onufry Wojtaszczyk.
\newblock Solving connectivity problems parameterized by treewidth in single
  exponential time.
\newblock In {\em 2011 IEEE 52nd Annual Symposium on Foundations of Computer
  Science}, pages 150--159. IEEE, 2011.

\bibitem{DellHMTW14}
Holger Dell, Thore Husfeldt, D{\'{a}}niel Marx, Nina Taslaman, and Martin
  Wahlen.
\newblock Exponential time complexity of the permanent and the tutte
  polynomial.
\newblock {\em {ACM} Trans. Algorithms}, 10(4):21:1--21:32, 2014.
\newblock \href {https://doi.org/10.1145/2635812} {\path{doi:10.1145/2635812}}.

\bibitem{dyer2000complexity}
Martin Dyer and Catherine Greenhill.
\newblock The complexity of counting graph homomorphisms.
\newblock {\em Random Structures \& Algorithms}, 17(3-4):260--289, 2000.

\bibitem{focke2022counting}
Jacob Focke, D{\'a}niel Marx, and Pawe{\l} Rz{\k{a}}{\.z}ewski.
\newblock Counting list homomorphisms from graphs of bounded treewidth: tight
  complexity bounds.
\newblock In {\em Proceedings of the 2022 Annual ACM-SIAM Symposium on Discrete
  Algorithms (SODA)}, pages 431--458. SIAM, 2022.

\bibitem{FominLPS16}
Fedor~V. Fomin, Daniel Lokshtanov, Fahad Panolan, and Saket Saurabh.
\newblock Efficient computation of representative families with applications in
  parameterized and exact algorithms.
\newblock {\em J. {ACM}}, 63(4):29:1--29:60, 2016.
\newblock \href {https://doi.org/10.1145/2886094} {\path{doi:10.1145/2886094}}.

\bibitem{greenhill2000complexity}
Catherine Greenhill.
\newblock The complexity of counting colourings and independent sets in sparse
  graphs and hypergraphs.
\newblock {\em Computational Complexity}, 9(1):52--72, 2000.

\bibitem{ImpagliazzoP01}
Russell Impagliazzo and Ramamohan Paturi.
\newblock On the complexity of k-sat.
\newblock {\em J. Comput. Syst. Sci.}, 62(2):367--375, 2001.
\newblock \href {https://doi.org/10.1006/jcss.2000.1727}
  {\path{doi:10.1006/jcss.2000.1727}}.

\bibitem{jaeger1990computational}
Fran{\c{c}}ois Jaeger, Dirk~L Vertigan, and Dominic~JA Welsh.
\newblock {On the computational complexity of the Jones and Tutte polynomials}.
\newblock In {\em Mathematical Proceedings of the Cambridge Philosophical
  Society}, volume 108, pages 35--53. Cambridge University Press, 1990.

\bibitem{jansen2019computing}
Bart~MP Jansen and Jesper Nederlof.
\newblock Computing the chromatic number using graph decompositions via matrix
  rank.
\newblock {\em Theoretical Computer Science}, 795:520--539, 2019.

\bibitem{kazeminia2019counting}
Amirhossein Kazeminia and Andrei~A Bulatov.
\newblock Counting homomorphisms modulo a prime number.
\newblock {\em arXiv preprint arXiv:1905.10682}, 2019.

\bibitem{Koivisto06}
Mikko Koivisto.
\newblock An $o^*(2^n)$ algorithm for graph coloring and other partitioning
  problems via inclusion--exclusion.
\newblock In {\em 47th Annual {IEEE} Symposium on Foundations of Computer
  Science {(FOCS} 2006), 21-24 October 2006, Berkeley, California, USA,
  Proceedings}, pages 583--590. {IEEE} Computer Society, 2006.
\newblock \href {https://doi.org/10.1109/FOCS.2006.11}
  {\path{doi:10.1109/FOCS.2006.11}}.

\bibitem{LokshtanovMS18}
Daniel Lokshtanov, D{\'{a}}niel Marx, and Saket Saurabh.
\newblock Known algorithms on graphs of bounded treewidth are probably optimal.
\newblock {\em {ACM} Trans. Algorithms}, 14(2):13:1--13:30, 2018.
\newblock \href {https://doi.org/10.1145/3170442} {\path{doi:10.1145/3170442}}.

\bibitem{Nederlof20}
Jesper Nederlof.
\newblock {Bipartite {TSP} in $O(1.9999^n)$ time, assuming quadratic time
  matrix multiplication}.
\newblock In Konstantin Makarychev, Yury Makarychev, Madhur Tulsiani, Gautam
  Kamath, and Julia Chuzhoy, editors, {\em Proccedings of the 52nd Annual {ACM}
  {SIGACT} Symposium on Theory of Computing, {STOC} 2020, Chicago, IL, USA,
  June 22-26, 2020}, pages 40--53. {ACM}, 2020.
\newblock \href {https://doi.org/10.1145/3357713.3384264}
  {\path{doi:10.1145/3357713.3384264}}.

\bibitem{oxley1979tutte}
James~G Oxley, Dominic~JA Welsh, et~al.
\newblock The tutte polynomial and percolation.
\newblock {\em Graph Theory and Related Topics}, pages 329--339, 1979.

\bibitem{Valiant06}
Leslie~G. Valiant.
\newblock Accidental algorithms.
\newblock In {\em 47th Annual {IEEE} Symposium on Foundations of Computer
  Science {(FOCS} 2006), 21-24 October 2006, Berkeley, California, USA,
  Proceedings}, pages 509--517. {IEEE} Computer Society, 2006.
\newblock \href {https://doi.org/10.1109/FOCS.2006.7}
  {\path{doi:10.1109/FOCS.2006.7}}.

\bibitem{RooijBR09}
Johan M.~M. van Rooij, Hans~L. Bodlaender, and Peter Rossmanith.
\newblock Dynamic programming on tree decompositions using generalised fast
  subset convolution.
\newblock In Amos Fiat and Peter Sanders, editors, {\em Algorithms - {ESA}
  2009, 17th Annual European Symposium, Copenhagen, Denmark, September 7-9,
  2009. Proceedings}, volume 5757 of {\em Lecture Notes in Computer Science},
  pages 566--577. Springer, 2009.
\newblock \href {https://doi.org/10.1007/978-3-642-04128-0\_51}
  {\path{doi:10.1007/978-3-642-04128-0\_51}}.

\end{thebibliography}

\appendix
\section{Proofs omitted from Section \ref{sec:collow}}
\label{sec:collowomit}
\subsection{Proof of Theorem \ref{thm:ColGad}}
We first reduce the lists of each $b_i$ to $\{1,2\}$ using the following gadget.
\begin{lemma} \label{lem:ColReduc}
Let $q,k \in \mathbb{N}$ and let $a\in [q]$. There is a graph $G$ with $b,b'\in V(G)$ and color lists $L_v\subseteq [q]$ for $v\in V(G)$, such that $L_b=[q]$ and the following two properties hold:
	\begin{itemize}
		\item for all $c_b\in [q]$, there is a unique list coloring $c$ of $G$ with $c(b)=c_b$,
		\item for all list colorings $c$ of $G$, if $c(b) = a$, then $c(b') = 1$ and if $c(b)\neq a$ then $c(b')=2$.
	\end{itemize}
\end{lemma}
\begin{proof}
We first note that it is easy to `relabel colors', as shown in the construction\footnote{If $a=1$ or $a'=2$ we slightly change the construction by removing the top left or top right vertex respectively.} in Figure \ref{fig:lb4}.
	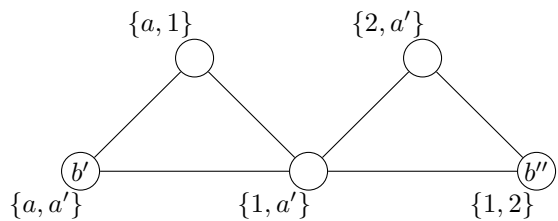
\begin{figure}[h]
		\centering
		\begin{tikzpicture}[-, scale = 1.5]
		\tikzstyle{every state} = [inner sep = 0.1mm, minimum size = 5mm]
		
		\node[state]	(BP)	at (0,0)				{$b'$};
		\node[]			(LB)	at ($(BP)+(-0.3,-0.3)$)	{$\{a,a'\}$};
		\node[state]	(V1)	at (1,1)				{};
		\node[]			(L1)	at ($(V1)+(-0.3,0.3)$)	{$\{a,1\}$};
		\node[state]	(V2)	at (2,0)				{};
		\node[]			(L2)	at ($(V2)+(-0.3,-0.3)$)	{$\{1,a'\}$};
		\node[state]	(V3)	at (3,1)				{};
		\node[]			(L3)	at ($(V3)+(-0.3,0.3)$)	{$\{2,a'\}$};
		\node[state]	(V4)	at (4,0)				{$b''$};
		\node[]			(L4)	at ($(V4)+(-0.3,-0.3)$)	{$\{1,2\}$};
		
		\path	(BP)	edge	node	{}	(V1)
				(V2)	edge	node	{}	(BP)
						edge	node	{}	(V1)
						edge	node	{}	(V3)
						edge	node	{}	(V4)
				(V3)	edge	node	{}	(V4);
				
		\end{tikzpicture}
		\caption{A gadget to `relabel colors'. It has two special vertices $b'$ and $b''$, and lists are depicted with sets. For any list coloring $c$ of the depicted gadget, if $c(b') = a$, then $c(b'') = 1$ and if $c(b') = a'$, then $c(b'') = 2$. In both cases, there is a unique way to color the remaining vertices.}
		\label{fig:lb4}
	\end{figure}
We can therefore first make a gadget for which $b'$ has color list $\{a,a'\}$ for some $a'\neq a$, and then relabel $a,a'$ to $1,2$. By symmetry, we can therefore assume that $a=1$ (or simply replace $1$ with $a$ and $2$ with $a'$ in the argument below).
Let
	\[
	V = \{b, b'\} \cup \{s_i : i = 2, \dots, q\} \cup \{t_i : i = 2, \dots, q\}
	\]
	and
	\[
	E = \{s_ib : i = 2, \dots, q\} \cup \{s_ib' : i = 2, \dots, q\} \cup \{s_it_j : i,j = 2, \dots, q\}.
	\]
	Now let $L_b=[q]$, $L_{b'} = \{1,2\}$ and $L_{t_i} = L_{s_i} = \{1,i\}$ for $i\in \{2,\dots,q\}$. A depiction is given in Figure \ref{fig:lb1}.
	\begin{figure}[h]
		\centering
		\begin{tikzpicture}[-, scale = 1.5]
		\tikzstyle{every state} = [inner sep = 0.1mm, minimum size = 5mm]
		
		\node[state]	(B)		at (0,0)		{$b$};
		\node[]			(LB)	at ($(B)+(-0.3,-0.3)$)	{$\{1, \dots, q\}$};
		\node[state]	(S2)	at (1,1)		{$s_2$};
		\node[]			(LB)	at ($(S2)+(-0.5,0)$)	{$\{1,2\}$};
		\node[state]	(SQ)	at (3,1)		{$s_q$};
		\node[]			(LB)	at ($(SQ)+(0.5,0)$)		{$\{1,q\}$};
		\node[state]	(T2)	at (1,2)		{$t_2$};
		\node[]			(LB)	at ($(T2)+(-0.5,0)$)	{$\{1,2\}$};
		\node[state]	(TQ)	at (3,2)		{$t_q$};
		\node[]			(LB)	at ($(TQ)+(0.5,0)$)		{$\{1,q\}$};
		\node[state]	(BP)	at (4,0)		{$b'$};
		\node[]			(LB)	at ($(BP)+(-0.3,-0.3)$)	{$\{1,2\}$};
		
		\node			(DS)	at ($(S2)!0.5!(SQ)$)	{$\dots$};
		\node			(DT)	at ($(T2)!0.5!(TQ)$)	{$\dots$};
		
		\path	(B)		edge	node	{}	(S2)
						edge	node	{}	(SQ)
						edge	node	{}	($(DS)!0.4!(B)$)
				(BP)	edge	node	{}	(S2)
						edge	node	{}	(SQ)
						edge	node	{}	($(DS)!0.4!(BP)$)
				(T2)	edge	node	{}	(S2)
						edge	node	{}	(SQ)
						edge	node	{}	($(DS)!0.5!(T2)$)
				(TQ)	edge	node	{}	(S2)
						edge	node	{}	(SQ)
						edge	node	{}	($(DS)!0.5!(TQ)$);
		
		\end{tikzpicture}
		\caption{The construction of the list coloring instance of the proof of Lemma \ref{lem:ColReduc}.}
		\label{fig:lb1}
	\end{figure}
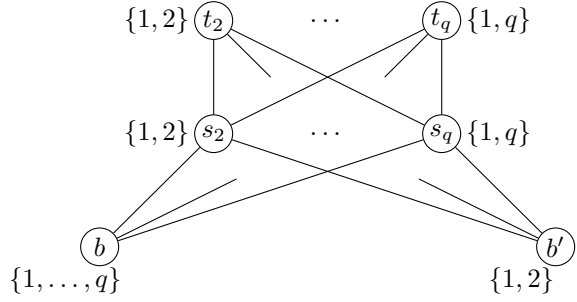
	
    If a list coloring $c$ of $G$ satisfies $c(b) = 1$, then $c(s_i) = i$ and thus $c(t_i) = 1$ for each $i\in \{2,\dots,q\}$. In particular $c(s_2) = 2$ and $c(b') = 1$. 
	
	If $c_b\in \{2,\dots,q\}$, then any list coloring $c$ with $c(b)=c_b$ satisfies $c(s_i) = 1$ and $c(t_i) = i$ for all $i\in \{2,\dots,q\}$, and so $c(b') = 2$.
	
	This proves that, starting with the color $c_b\in [q]$ for $b$, there is always a unique extension to a list coloring of $G$, and this satisfies the property that vertex $b'$ receives color $1$ if $c_b=1$, and receives color $2$ otherwise.
\end{proof}

We also make use of the following construction.

\begin{lemma} \label{lem:ColAtom}
Let $k,\ell \in \mathbb{N}$. There is a graph $G$, a subset of vertices $B = \{b_1, \dots, b_k\}\subseteq V$ of size $k$, and color lists $L_v$ for all $v\in V(G)$ such that:
\begin{itemize}
    \item $L_{b_i} = \{1,2\}$ for all $i\in \{1,\dots,k\}$,
    \item there are exactly $\ell$ list colorings $c$ of $G$ with $c(B)=\{1\}$,
    \item for each partial coloring $c_B$ of $B$ with $c_B(B)\neq \{1\}$, there is a unique extension of $c_B$ to a list coloring of $G$.
\end{itemize}
\end{lemma}

\begin{proof}
	We start with $V = B$ and add a path\footnote{When $\ell = 1$, we add no vertices of the form $w_i$ and the statements of the lemma immediately follow.} $w_1, \dots, w_{\ell-1}$ with color lists
	\[
	L_{w_i} =\begin{cases}
	\{2, 3\}    & \text{if } i \equiv_3 1, \\
	\{1, 3\}    & \text{if } i \equiv_3 2, \\
	\{1, 2\}    & \text{if } i \equiv_3 0, \\
	\end{cases}
	\]
	and add edges $b_iw_1$ for $i = 1, \dots, k$.  A depiction is given in Figure \ref{fig:lb2}.
	\begin{figure}[h]
		\centering
		\begin{tikzpicture}[-, scale = 2]
		\tikzstyle{every state} = [inner sep = 0.1mm, minimum size = 8mm]
		\node[state]	(B1)	at (0,0)	{$b_1$};
		\node[]			(LB1)	at ($(B1)+(-0.3,0.3)$)	{$\{1,2\}$};
		\node[state]	(BK)	at (2,0)	{$b_k$};
		\node[]			(LBK)	at ($(BK)+(-0.3,0.3)$)	{$\{1,2\}$};
		\node[state]	(W1)	at (1,-1)	{$w_1$};
		\node[]			(LB1)	at ($(W1)+(-0.3,-0.3)$)	{$\{2,3\}$};
		\node[state]	(W2)	at (2,-1)	{$w_2$};
		\node[]			(LB1)	at ($(W2)+(-0.3,-0.3)$)	{$\{1,3\}$};
		\node[state]	(W3)	at (3,-1)	{$w_3$};
		\node[]			(LB1)	at ($(W3)+(-0.3,-0.3)$)	{$\{1,2\}$};
		\node[state]	(WL)	at (5,-1)	{$w_{\ell-1}$};
		\node[]			(LB1)	at ($(WL)+(-0.3,-0.3)$)	{$\{2,3\}$};
		
		\node			(DB)	at ($(B1)!0.5!(BK)$)	{$\dots$};
		\node			(DW)	at ($(W3)!0.5!(WL)$)	{$\dots$};
		
		\path	(W1)	edge	node	{}	(B1)
						edge	node	{}	($(W1)!0.6!(DB)$)
						edge	node	{}	(BK)
						edge	node	{}	(W2)
				(W3)	edge	node	{}	(W2)
						edge	node	{}	($(W3)!0.6!(DW)$)
				(WL)	edge	node	{}	($(WL)!0.6!(DW)$);
		\end{tikzpicture}
		\caption{Construction in the proof of Lemma \ref{lem:ColAtom} when $\ell \equiv_3 2$.}
		\label{fig:lb2}
	\end{figure}
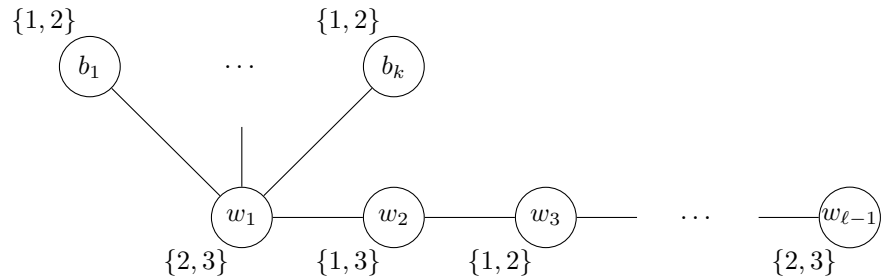
	
	If a list coloring $c$ satisfies $c(b_i) = 2$ for some $i\in [k]$, then $c(w_1) = 3$, $c(w_2) = 1$, $c(w_3) = 2$ etcetera. Hence there is a unique extension of any partial coloring of $B$ that assigns color 2 somewhere.
	
	If $c(b_i) = 1$ for all $i\in [k]$, then we have a choice for the color of $w_1$. If $c(w_1) = 3$ then we get the same propagation as before, however if $c(w_1) = 2$, then we have a choice for the color of $w_2$. Using a simple induction argument we find that the number of possible list colorings with $c(B)=1$ equals $\ell$.
\end{proof}

We are now ready to construct the main gadget.
\begin{proof}[Proof of Theorem \ref{thm:ColGad}]
Let $f\in [q]^k$.
We create vertices $b_1,\dots,b_k$ and give them all $[q]$ as list.

For each partial coloring $\alpha \in [q]^k$, we create a graph $G_{\alpha}$ that contains $b_1,\dots,b_k$ in their vertex set, but the graphs are on disjoint vertex sets otherwise (we `glue' the graphs on the special vertices $b_1,\dots,b_k$). It suffices to show that we can find lists for the `private' vertices of $G_{\alpha}$ such that the number of extensions of a coloring $c_B$ of $B$ is 1 if $c_B(b_i)\neq \alpha(i)$ for some $i\in [k]$, and $f(\alpha)$ otherwise. The resulting gadget will then have $1\cdot 1 \cdot \ldots \cdot 1 \cdot f(\alpha)=f(\alpha)$ possible extensions for the precoloring $\alpha$, as desired.

We now turn to constructing the gadget $G_\alpha$ for a fixed coloring $\alpha\in [q]^k$. We first reduce to the case in which each $b_i$ has $\{1,2\}$ as list. Let $i\in [k]$. Using Lemma \ref{lem:ColReduc} with $a=\alpha(i)$, we obtain a gadget $H_{b,b'}$ and identify the special vertex $b$ with $b_i$. For each $\alpha$, we obtain a new set of vertices $b_1',\dots,b_k'$ with lists $\{1,2\}$. We then glue these onto the special vertices from a gadget obtained by applying Lemma \ref{lem:ColAtom} with $\ell = f(\alpha)$. If $b_1,\dots,b_k$ are colored as specified by $\alpha$, then $b_1',\dots,b_k'$ all receive color 1 and $G_\alpha$ has $f(\alpha)$ possible extensions; however if some $b_i$ receives the wrong color, the corresponding $b_i'$ receives color $2$ and there is a unique extension to the rest of $G_\alpha$.
	
It remains to show the bounds on the number of vertices and the cutwidth. We give the very rough upperbound of $6kq^{k+2}$ on the cutwidth. The gadget from Lemma \ref{lem:ColReduc} has cutwidth at most $q^2+6$ (since this is an upper bound on the number of edges in that construction). The gadgets from Lemma \ref{lem:ColAtom} have cutwidth at most $k$. A final cut decomposition can be obtained by first enumerating the vertices in $B$, and then adding the cut decompositions of each $G_\alpha$, one after the other. 

Finally the number of vertices of the graph is upper bounded by $q^k$ times the maximum number of vertices of the graph $G_{\alpha}$. The gadget of Lemma \ref{lem:ColReduc} has at most $2q+6$ vertices and there are $k$ of them, so they contribute at most $12kq$ vertices. The gadgets from Lemma \ref{lem:ColAtom} add at most $f(\alpha)$ vertices. In total, 
	\[
	|V(G_f)| \leq 20kq^{k+1}\max(f). \qedhere
	\]
\end{proof}

\subsection{Remaining details of the proof of Theorem \ref{thm:ColRed}}
We first describe the gadgets for the `color transfer' (the first desired property).
Let $j\in [m-1]$ and $i\in [n]$. We will apply Theorem \ref{thm:ColGad} to a function $f_{i,j}$ with boundary set $B_{i,j}=(s_{i,j},t_{i,j})$ and $\max(f_{i,j})=p$, resulting in a graph on $O_{p,q}(1)$ vertices. Let $J_1$ be the $q\times q$  coloring compatibility matrix of a single edge, and let $J_{1}^{-1}$ denote its inverse over $\F_p$ (that is, $J_1^{-1}J_1\equiv_p I_q$, the $q\times q$ identity matrix). We `choose a representative' $\widetilde{J_1}$, which has entries in $\{1,\dots, p\}$ that are equivalent to those in $J_1^{-1}$ modulo $p$. 
For $c_1,c_2\in [q]$ possible colors for $s_{i,j}$ and $t_{i,j}$ respectively, we set
\[
f_{i,j}(c_1,c_2) = \widetilde{J_1}[c_1,c_2].
\]
Let $V_{i,j}$ denote the vertices in the gadget obtained by applying Theorem \ref{thm:ColGad} to $(f_{i,j},B_{i,j})$ that are not in $B_{i,j}$. Let $c_1,c_3\in [q]$. The number of list colorings  $c$ of the graph induced on $B_{i,j}\cup V_{i,j}\cup \{s_{i,j+1}\}$ with $c(s_{i,j})=c_1$ and $c(s_{i,j+1})=c_3$ is equal to
\[
\sum_{c_2\in [q]}f_{i,j}(c_1,c_2) J_1[c_2,c_3] =(\widetilde{J_1}J_1)[c_1,c_3],
\]
since for any coloring $c_2$ we have $f_{i,j}(c_1,c_2) J_1[c_2,c_3]$ such colorings with $c(t_{i,j}) = c_2$, by definition of $f_i,j$ and $J_1$. Therefore, modulo $p$ this number of extensions is equal to $1$ if $c_1=c_3$ and $0$ otherwise, as desired.

We now describe the gadgets that check the constraints. Let $j\in [m]$ and let $i_1,\dots,i_\ell$ be given so that the $j$th constraint only depends on the variables $x_{i_1},\dots,x_{i_\ell}$ (where by assumption $\ell\leq r$). We will apply Theorem \ref{thm:ColGad} to a function $g_{j}$ with boundary set $B_{j}=(s_{i_1,j},\dots,s_{i_\ell ,j})$ and $\max(g_{j})=p$, resulting in a graph on $O_{p,q,\ell}(1)=O_{p,q,r}(1)$ vertices. We set $g_j(c_{1},\dots,c_{\ell})$ to be equal to 1 if the assignment $(c_1,\dots,c_{\ell})$ to $(x_{i_1},\dots,x_{i_\ell})$ satisfies the $j$th constraint, and $p$ otherwise. This ensures the second property described in the proof of Theorem \ref{thm:ColRed}.

\subsection{Proof of Theorem \ref{thm:main_upperboundcse}, upper bound}
\label{sec:upperboundcse}
Let $p$ be a prime. In this section we show that there is an algorithm that counts the number of connected spanning edge sets modulo $p$ in an $n$-vertex graph of treewidth at most $\tw$ in time $p^{\tw}n^{O(1)}$. The algorithm uses dynamic programming on the tree decomposition, reminiscent of cut-and-count. We will first need some additional notions about tree decompositions.

Let $G$ be an $n$-vertex graph with a tree decomposition $(T,(B_x)_{x\in V(T)})$ such that $T$ is a rooted tree and $|B_x|\leq \tw+1$ for all $x\in V(T)$.  Since $T$ is rooted, we can consider the children and  descendants of $x$ in $T$.
In polynomial time, we can adjust the tree decomposition to an edge-introduce tree decomposition without increasing the width. For an  \emph{edge-introduce}  tree decomposition (called a nice tree decomposition in \cite{cygan2011solving}) all bags $B_x$ are of one of the following types:
\begin{itemize}
\item Leaf bag: $x$ is a leaf of $T$.
\item Introduce vertex bag: $x$ is an internal vertex of $T$ with a single child $y$ for which $B_x = B_y \cup \{v\}$ for some vertex $v\in V(G)\setminus B_y$.
\item Introduce edge bag: $x$ is an internal vertex of $T$ with a single child $y$ with $B_x=B_y$, and $x$ is moreover labelled by an edge $uv \in E(G)$ with $u,v\in B_y$. Every edge in $E$ is `introduced' (used as label) exactly once.
\item Forget bag: $x$ is an internal vertex of $T$ with a single child $y$ for which $B_x = B_y \setminus \{v\}$ for some $v\in B_y$.
\item Join bag: $x$ is an internal vertex of $T$ with two children $y$ and $z$ for which $B_x=B_y=B_z$.
\end{itemize} 
We write $G_x$ for the graph on the vertices in the bags that are a descendant of $x$ in $T$, which has as edge set those edges that have been `introduced' by some bag below $B_x$.

\begin{proof}[Proof of Theorem \ref{thm:main_upperboundcse}]
By the discussion above, we assume that an edge-introduce tree decomposition $(T,(B_x)_{x\in V(T)})$ for $G$ has been given.
For each $x \in V(T)$ and each partition $X_1,\dots,X_p$ of $B_x$, we define
\begin{align*}
T_x[X_1,\dots,X_p] =|\{(X,V_1,\dots,V_p)\::\:& V_1\sqcup \dots \sqcup V_p= V(G_x),\\
												&X\subseteq E(G_x) \setminus (\cup_{i<j} E(V_i,V_j)) \}|,
\end{align*}
where $\sqcup$ denotes the disjoint union (that is, we take the union and assume the sets are disjoint).
Thus, we count the number of edge sets $X$ of $G_x$ and number of vertex partitions $P$ of $G_x$ such that $X$ does not `cross' $P$ and $P$ `induces' the partition $(X_1,\dots,X_p)$ on $B_x$. Note that this is not done modulo $p$, but over the reals, and also that we do not put in the `connected spanning' constraints.

We claim that for $r$ the root of $T$, 
\[
\frac1p \sum_{R_1,\dots,R_p} T_r[R_1,\dots,R_p]
\]
equals the number of connected spanning edge sets of $G$ modulo $p$, where the sum is over the possible partitions of $B_r$. Note that $\sum_{R_1,\dots,R_p} T_r[R_1,\dots,R_p]$ counts the number of tuples $(X,V_1,\dots,V_p)$ such that $X$ is an edge set of $G$ that does not cross the partition $(V_1,\dots,V_p)$.  For an edge set $X$ of $G$, let $U_1,\dots, U_k$ denote the connected components (which may be single vertices) of $G[X]$, the graph with vertex set $V(G)$ and edge set $X$. The number of partitions $(V_1,\dots,V_p)$ of $G$ for which $X$ respects the partition equals $p^k$: we assign each connected component an element of $[p]$. Therefore,  the contribution of each $X$ to $\sum_{R_1,\dots,R_p} T_r[R_1,\dots,R_p]$ is divisible by $p$, and is divisible by $p^2$ if and only if $G[X]$ is not connected. This proves the claim since $G[X]$ is connected if and only if $X$ is connected and spanning.

What remains to show is that we can calculate the table entries in the claimed time complexity. If $x \in V(T)$ is a leaf, then we can calculate the table entry naively. We now assume that all strict descendants of $x$ already have their table entry calculated. We consider two cases.

Suppose first that $B_x$ is a `join bag': $x$ has two children $y$ and $z$ with $B_x=B_y=B_z$, and $E(G_y)\cap E(G_y)=\emptyset$. For any partition $(X_1,\dots,X_p)$ of $B_x$, we can also consider this as a partition of the bags $B_y$ and $B_z$. We claim that
\[
T_x[X_1,\dots,X_p] = T_y[X_1,\dots,X_p]\cdot T_z[X_1,\dots,X_p].
\]
This follows from the fact that any $(V_1,\dots,V_p,X)$ counted for $T_x[X_1,\dots,X_p]$, is determined by the following parts: the partition $X_1,\dots,X_p$ of $B_x$ (which is `fixed'), the remaining partition of the vertices in 
\[
V(G_x)\setminus V(B_x)=(V(G_y)\setminus V(B_x))\sqcup (V(G_z)\setminus V(B_x))
\]
(where we obtain a disjoint union by the standard tree decomposition properties) and the edge set $X\subseteq E(G_x)=E(G_y)\sqcup E(G_z)$ (where we find a disjoint union by our additional `edge-introduce' property).

Suppose now that $x$ has a single child $y$. Fix a partition $P_x=(X_1,\dots,X_p)$ for $B_x$ and consider a partition $P_{y}=(Y_1,\dots,Y_p)$ of $B_y$ that agrees with the partition $P_x$, in the sense that $X_i\cap Y_j=\emptyset$ for $i\neq j$. Note that since $|B_y \setminus B_x|\leq 1$, there are at most $p$ partitions $P_y$ that agree with $P_x$. We can therefore efficiently compute 
\begin{equation}
\label{eq:cseupp}
C \sum_{P_{y}:P_{y}\sim P_x} T_{y}[Y_1,\dots,Y_p]
\end{equation}
where the sum is over those $P_{y}=(Y_1,\dots,Y_p)$ that agree with the partition $P_x=(X_1,\dots,X_p)$, and where $C=C(X_1,\dots,X_p,y)$ is given by the number of possible subsets of $E(G_x)\setminus E(G_{y})$ that respect the partition $P_x$ (always counting the empty set). We will now show that (\ref{eq:cseupp}) equals $T_x[X_1,\dots,X_p]$.

When we fix a partition $(V_1,\dots,V_p)$ of $V(G_x)$ agreeing with $P_x$, then this uniquely defines a partition $P_y$ agreeing with $P_x$, namely $Y_i=V_i\cap B_y$ for all $i\in [p]$. Similarly, any partition $(V'_1,\dots,V_p')$ of $V(G_y)$ agreeing with $P_y$ (for some $P_y$ agreeing with $P_x$), uniquely defines a partition $(V_1,\dots,V_p)$ of $V(G_x)$ that agrees with $P_x$ (where we may need to add a vertex from $B_x\setminus B_y$ in the place specified by $P_x$). 

If $x$ is not introducing an edge, then $E(G_x)=E(G_y)$. Therefore, there is a one-to-one correspondence between tuples $(V_1,\dots,V_p,X)$ counted for $T_x$ and tuples $(P_y,V_1\cap B_y,\dots,V_p\cap B_y,X)$ counted for  (\ref{eq:cseupp}). Since in this situation $C=1$,  (\ref{eq:cseupp}) indeed equals $T_x[X_1,\dots,X_p]$.

If $x$ is introducing the edge $uv\in E(G_x)\setminus E(G_y)$ and $uv$ respects the partition $P_x$, then there is a two-to-one correspondence between tuples $(V_1,\dots,V_p,X')$ and $(V_1,\dots,V_p,X'\cup \{uv\})$ counted for $T_x$ and tuples $(P_y,V_1\cap B_y,\dots,V_p\cap B_y,X')$ counted for  (\ref{eq:cseupp}), where $X'\subseteq E(G_y)$. In this case $C=2$. If $uv$ does not respect the partition, there is a one-to-one correspondence again and $C=1$. This proves that  (\ref{eq:cseupp}) equals $T_x[X_1,\dots,X_p]$ in all cases.

The number of partitions $(X_1,\dots,X_p)$ of a bag $B_x$ is at most $p^{|B_x|}$ (assign each vertex of $B$ to an element of $[p]$). The computation of a single table entry is always done in time polynomial in $n$, and the size of the tree is also polynomial in $n$. Therefore, the running time of this algorithm is at most $p^{\tw}n^{O(1)}$.
\end{proof}
\end{document}